\documentclass[12pt]{amsart}
\usepackage{graphicx}

\usepackage{graphicx,tikz,wasysym,tikz-cd}
\usepackage{xcolor}
\usepackage{graphicx}
\usepackage{epsfig}

     \usetikzlibrary{arrows,decorations.markings,shapes,positioning,fit,shadows,petri,backgrounds}
\usepackage{color}

\newtheorem{thm}{Theorem}[section]

\newtheorem{defn}[thm]{Definition}
\newtheorem{lemma}[thm]{Lemma}
\newtheorem{conj}[thm]{Conjecture}

\newtheorem{remark}[thm]{Remark}

\usepackage{amsmath}
\usepackage{amsxtra}
\usepackage{amscd}
\usepackage{amsthm}
\usepackage{amsfonts}
\usepackage{amssymb}
\usepackage{eucal}
\newcommand{\bmb}{\left( \begin{array}{rr}}
\newcommand{\enm}{\end{array}\right)}

\newcommand{\E}{\mathcal E}

\newcommand{\g}{{\mathfrak{g}}}

\newcommand{\gl}{{\mathfrak{gl}}}

\newcommand{\bX}{{\mathbf X}}

\renewcommand{\sl}{{\mathfrak{sl}}}
\newcommand{\KR}{{{\rm KR}}}

\newcommand{\Z}{{\mathbb Z}}
\newcommand{\Q}{{\mathcal Q}}
\newcommand{\R}{{\mathbb R}}
\newcommand{\N}{{\mathbb N}}

\newcommand{\ba}{{\mathbf a}}

\newcommand{\bn}{{\mathbf n}}

\newcommand{\bx}{{\mathbf x}}

\newcommand{\al}{{\alpha}}

\newcommand{\half}{\frac12}
\numberwithin{equation}{section}
\numberwithin{table}{section}
\tikzset{>=latex}

\begin{document}

\title{Macdonald operators and quantum Q-systems for classical types}
\author{Philippe Di Francesco} 
\address{Department of Mathematics, University of Illinois MC-382, Urbana, IL 61821, 
U.S.A. e-mail: philippe@illinois.edu and 
Institut de Physique Th\'eorique du Commissariat \`a l'Energie Atomique, 
Unit\'e de Recherche associ\'ee du CNRS,
CEA Saclay/IPhT/Bat 774, F-91191 Gif sur Yvette Cedex, 
FRANCE. e-mail: philippe.di-francesco@cea.fr
}
\author{Rinat Kedem} 
\address{Department of Mathematics, University of Illinois MC-382, Urbana, IL 61821, 
U.S.A. e-mail: rinat@illinois.edu}

\begin{abstract}
We propose solutions of the quantum Q-systems of types $B_N,C_N,D_N$ in terms of 
$q$-difference operators, generalizing our previous construction for the Q-system of type $A$. The difference operators are interpreted as $q$-Whittaker limits of discrete time evolutions of  Macdonald-van Diejen type operators.
We  conjecture that these new operators act as raising and lowering operators for $q$-Whittaker 
functions, which are special cases of graded characters of fusion products of KR-modules.
\end{abstract}

\maketitle

\date{\today}

\centerline{\it To Nicolai Reshetikhin on his 60th birthday}

\section{Introduction}

The characters of  tensor products of KR-modules of Yangians, quantum affine algebras or affine algebras have fermionic formulas, generalizing Bethe's original counting formula of the Bethe eigenstates of the Heisenberg spin chain \cite{Bethe}. The fermionic formulas, in the case of KR-modules of Yangians of the classical Lie algebras $\g=ABCD$, were conjectured by Kirillov and Reshetikin \cite{KR} and further generalized in \cite{HKOTY}. They were proved for the case of any simple Lie algebra $\g$ in \cite{AK,DFK08}. In the course of the proof, it becomes clear that there is a close connection between solutions of recursion relations known as the Q-systems and the fermionic formulas.

In \cite{KR,HKOTY}, $q$-analogs of the fermionic characters of tensor products of KR-modules, or graded characters, are also presented. There are several interpretations of the grading of these tensor products,
which gives rise to these $q$-analogues, or graded characters: (1) as linearized energy function for the
corresponding generalized inhomogeneous Heisenberg spin chain, (2) as a charge function for the crystal limit of the corresponding quantum affine algebra modules, when it exists, or (3) as the natural grading 
 of the underlying affine algebra \cite{FL}. The latter definition was used in \cite{AK,qKR,Simon} to prove the $q$-graded fermionic character formulas.  It turns out that there is a close connection between the graded character formulas and $q$-deformed versions of Q-systems known as the quantum Q-systems.
 
Quantum Q-systems are defined by quantizing \cite{BZ} the cluster algebraic \cite{FZ} structure of the classical Q-systems \cite{Ke07,DFKQcluster}.  They are recursion relations for non-commuting variables $\{\Q_{a,k}\}$, with $a$ running over the Dykin labels of $\g$ and $k\in\Z$.

Graded characters are Weyl-symmetric functions with coefficients in $\Z_+[q]$. The relation with the quantum Q-system can be schematically described as follows. One can construct a linear functional $\phi$ from the ordered product of (opposite, with parameter $q^{-1}$) quantum Q-system solutions $\Q_{a,k}^*$ to the graded characters  of the corresponding 
tensor product: 
\begin{equation}\label{linfun}
\chi_{\bn}(q^{-1};\bx)= \phi\left( {\displaystyle \prod_{a,k}^{\rightarrow} (\Q^*_{a,k})^{n_{a,k}}}\right),
\end{equation}
where $\chi_\bn(q;\bx)$ is the graded character of the tensor product (or fusion product in the sense of \cite{FL}) of KR-modules $\otimes_{a,k} \KR_{k\omega_a}^{\otimes n_{a,k}}$, $\bn=\{n_{a,k}\}$ denoting the collection of tensor powers,
and $\bx=(x_1,x_2,...,x_N)$. Here, $\KR_{k\omega_a}$ is a $\g$-module, the restriction of the KR-module of the (quantum) affine algebra module, which has highest weight $k\omega_a$, $k\in\N$ and $\omega_a$ being one of the fundamental weights of $\g$. The arrow on top of the product sign refers to a specific ordering of the terms \cite{qKR,Simon}, in which long and short roots play different roles.

Tensoring the tensor product above by an extra factor $\KR_{k'\omega_b}$  corresponds to the insertion of the factor $\Q_{b,k'}^*$ on the right in the product inside the functional, if $k'$ is sufficiently large, so that the product remains ordered.

In the case of $\g=\sl_N$ we introduced \cite{DFK15,DFK16} a set of $q$-difference operators 
${\mathcal D}_{b,k'}$,  acting on the space of symmetric functions of $\bx$ to the right, and representing the insertion on the right of the factor $\Q^*_{b,k'}$ in the linear functional of equation \eqref{linfun}. 
The following diagram explains the action by difference operators on the space of symmetric functions:
\begin{center}
\begin{tikzcd}[column sep=large,row sep=large]
{\displaystyle \prod_{a,k}^{\rightarrow} (\Q_{a,k}^*)^{n_{a,k}}}
\arrow[r, "\Q^*_{b,k'}"]
\arrow[d, "\phi"] 
& {\displaystyle\left(\prod_{a,k}^{\rightarrow} (\Q_{a,k}^*)^{n_{a,k}}\right)\Q^*_{b,k'}}
\arrow[d, "\phi"] \\
\chi_{\bn}(q^{-1};\bx)
\arrow[r,  "{\mathcal D}_{b,k'}"]
& \chi_{\bn+\epsilon_{b,k'}}(q^{-1};\bx)
\end{tikzcd}
\end{center}
The result is an explicit expression for graded characters as the iterated action of 
$q$-difference operators ${\mathcal D}_{a,k}$ on the constant $1$.
The difference operators satisfy the quantum Q-system. We refer to them as the functional representation of the quantum Q-system. 

In \cite{DFK15}, it was observed that the difference operators $\mathcal D_{a,k}$ for $k=0$ are the $t\to \infty$ limit of
the Macdonald difference operators of type $A$, of which Macdonald polynomials form a set of common eigenfunctions.
Thus, in \cite{DFKqt}, we identified the  $t$-deformation of the $A$ type quantum Q-system as the spherical Double Affine Algebra Hecke (sDAHA) of type $A$. This is the algebra which underlies Macdonald theory \cite{macdo,Cheredbook}.

In the $t\to\infty$ limit, in which Macdonald polynomials tend to (dual) $q$-Whittaker functions with parameter $q^{-1}$, we further
identified in \cite{DFK15} the operator  representation ${\mathcal D}_{a,k}$ of the quantum Q-system when $k=1$ (resp. $k=-1$) as 
the $t\to\infty$ limits of the Kirillov-Noumi \cite{Kinoum} raising (resp. lowering) operators. These operators, acting on a Macdonald polynomial indexed by some Young diagram,
have the effect of adding (resp. subtracting) a column of $a$ boxes to the Young diagram. Equivalently, using the correspondence of the Young diagram $\lambda$ with a dominant $\gl_N$-weight, it corresponds to adding or subtracting the fundamental weight $\omega_a$.

As a consequence, in the case where all KR-modules in the tensor product are fundamental modules
($n_{a,k}=0$ unless $k=1$), the graded characters are identified as limits of Macdonald polynomials as $t\to\infty$ or $t=0$ upon changing $q\to q^{-1}$, and can therefore be identified with specialized $q$-Whittaker functions. See also \cite{LNSSS}.

We remark that the above difference operators can be compared to the so-called ``minuscule monopole operators" representing
the Coulomb branches of 4D $N=4$ quiver gauge theories \cite{BFN} in the particular case of the ``Jordan quiver" with a single node, and when the equivariant parameter $t$ is taken to infinity (with the result of imposing that the representation $N$ is trivial).

It is natural to look for the generalization of the functional representation  of the   quantum Q-systems,  corresponding to the affine algebras of types $BCD$. These are described in terms of the root systems of the finite $BCD$ type.
 Motivated by the results in type $A$,
we expect the functional representation of the other type quantum Q-systems to involve
the  $t\to\infty$ limits of the corresponding $BCD$ type generalized Macdonald operators. 
In this paper, we present, without proof, a set of such difference operators. Our main conjecture  \ref{qqconj} is that these operators satisfy the relevant quantum Q-systems relations.

The construction of the $BCD$ type difference operators is best understood by thinking of the
quantum Q-systems as evolution equations in the discrete time variable $k\in \Z$ for the elements 
$\Q_{a,k}$. In type $A$, from the relation to sDAHA, we noted in \cite{DFKqt} that 
the  discrete time evolution $k\to k+1$ is  given by the adjoint action of a generator of the $SL_2(\Z)$ symmetry of  DAHA. The latter is also expressed as the adjoint action of the ``Gaussian" function 
$\gamma^{-1}$ of the variables $\bx$, where
\begin{equation}\label{gaussian}
\gamma\equiv\gamma(\bx)=e^{\sum_{i=1}^N \frac{{\rm Log}(x_i)}{2{\rm Log}(q)}} ,
\end{equation}
and such that ${\mathcal D}_{a,k+1} \propto \gamma^{-1} {\mathcal D}_{a,k} \gamma$.

The aim of the present paper is to present constructions of difference operator solutions to the
quantum Q-systems of types $BCD$, such that they coincide, at $k=0$, with the $t\to\infty$ limit of suitable Macdonald operators, 
and which have raising/lowering properties at $k=\pm 1$. To this end, we first identify the  
$t\to\infty$ limits
of suitable Macdonald-type operators in types $BCD$ by use of works of Macdonald and van Diejen \cite{macdoroot,vandiej,vandiejtwo}. 
Next, we construct their time evolution, by a suitable Gaussian conjugation. Our main result are the conjectures \ref{qqconj} and \ref{knconj} stating that (1) these
operators obey a renormalized version of the quantum Q-systems in types $BCD$ and (2) the operators at times $k=\pm 1$ act as raising/lowering operators on the corresponding $q$-Whittaker functions.

\bigskip

\noindent{\bf Acknowledgments.} RK and PDF acknowledge support by NFS grant DMS18-02044
PDF is partially supported by the Morris and Gertrude Fine endowment.
RK thanks the Institute Physique Th\'eorique of CEA/Saclay for hospitality.

\section{Q-systems and quantum Q-systems}

\subsection{Weights and Roots}\label{rootsec}

Let $\g$ be a Lie algebras of classical type, ${\mathfrak g}\in
\{A_{N-1},D_N,B_N,C_N\}$.
For each of these, we list in table \ref{data} the standard data of fundamental weights $\omega_a$, the simple roots $\alpha_a$, and the conditions on the non-increasing sequence $\lambda=(\lambda_1\geq\lambda_2\geq \cdots \geq \lambda_N)$ corresponding to dominant weights $\sum_a n_a \omega_a=\sum_i \lambda_i e_i$ ($n_a\in\Z_+$). 
In table \ref{data}, the set $\{e_a\}_{a=1}^N$ is the standard basis of $\R^N$ whereas $\{\hat{e}_a=e_a-\rho/N\}_{a=1}^{N-1}$ where $\rho$ is the sum over all the basis elements.
\begin{table}\label{data}
\begin{center}
\begin{tabular}{|c | c | c | c |}
\hline
Algebra & Fundamental weights $\omega_a$ & simple roots $\alpha_a$ &  $\lambda$ \\
[.1cm]
\hline
$A_{N-1}$ &$ \displaystyle{\omega_a = \sum_{i=1}^a \hat e_i,\ a\in[1,N-1]} $&
$  e_a- e_{a+1} $ & $\lambda_a\in\Z_+$\\
[.1cm]
\hline
$B_N$ & $\displaystyle{\omega_a=\left\{\begin{array}{ll}
\displaystyle{\sum_{i=1}^a e_i,}& a<N;\\ \displaystyle{\frac{1}{2} \sum_{i=1}^N e_i,}& a=N. \\\end{array}\right.} $
&$ \begin{array}{l}e_a-e_{a+1},\ a<N\\ e_N, \ a=N\end{array}$
 & $\begin{array}{c}
 \lambda_a\in\Z_+ \hbox{ for all $a$ or}\\
 \lambda_a\in  \Z_++\half \hbox{ for all }a\end{array}$\\
[.1cm]
 \hline
$C_N$ & $\omega_a = \displaystyle{\sum_{i=1}^a e_i},\ a\in[1,N]$
&
$ \displaystyle{\begin{array}{l} e_a-e_{a+1},\ a<N\\
2 e_N,\ a=N \end{array}}$
& $\lambda_a\in\Z_+$ \\
[.1cm]
\hline
$D_N$ &$\begin{array}{l}
\omega_a = \displaystyle\sum_{i=1}^a e_i, \ a<N-1\\
\omega_{N-1}=  \half(\omega_{N-2} + e_{N-1}-e_N)\\
\omega_N= \half(\omega_{N-2}+e_{N-1}+e_N)\end{array} $
& $\displaystyle{\begin{array}{l} e_a - e_{a+1},\ a<N\\ e_a+e_{a+1} ,\ a=N
\end{array}}$
&$\begin{array}{c}\lambda_a\in\Z \hbox{ for all $a$ or }\\
 \lambda_a\in \Z+\half \hbox{for all $a$},\\
\lambda_{N-1}\geq |\lambda_N| \geq 0.
\end{array}$
 \\
\hline
\end{tabular}
\end{center}
\vskip.5cm
\caption{Root and weight data for the classical Lie algebras}
\end{table}

We also denote by $t_a$ the integers $\frac{2}{||\al_a||^2}$, so that $t_a=1$ for long roots and $t_a=2$ for the short roots in types $BC$.

\subsection{The classical Q-systems for simply-laced affine $ABCD$}

The Q-systems  are recursion relations for the variables $\{Q_{a,k}\}$ where $a$ is a label in the Dynkin diagram, and $k$ is any integer. We list the Q-systems associated with types $ABCD$ \cite{KR,KNS}. The boundary condition $Q_{0,k}=0$ is assumed in all cases.
\begin{eqnarray}
 \g=A_{N-1}: & &
Q_{a,k+1}Q_{a,k-1}=Q_{a,k}^2-Q_{a+1,k}Q_{a-1,k}, \qquad (a\in [1,N-1]),\label{ANqsys}\\
&& \quad Q_{N,k}=1. \nonumber \\ \nonumber\\
 \g=B_{N}: &&
Q_{a,k+1}Q_{a,k-1}=Q_{a,k}^2-Q_{a+1,k}Q_{a-1,k} ,\qquad (a\in [1,N-2]),\nonumber\\
&& Q_{N-1,k+1}\,Q_{N-1,k-1}=(Q_{N-1,k})^2- Q_{N,2k}\, Q_{N-2,k},\label{BNqsys}\\
&& Q_{N,2k+1}\,Q_{N,2k-1}= (Q_{N,2k})^2 -(Q_{N-1,k})^2 ,\nonumber\\
&& Q_{N,2k+2}\,Q_{N,2k}= (Q_{N,2k+1})^2 -Q_{N-1,k+1}\,Q_{N-1,k}. \nonumber\\ \nonumber\\
 \g=C_{N}: &&
Q_{a,k+1}Q_{a,k-1}=Q_{a,k}^2-Q_{a+1,k}Q_{a-1,k}, \qquad (a\in [1,N-2]),\nonumber\\
&& Q_{N-1,2k+1}Q_{N-1,2k-1}=Q_{N-1,2k}^2 -  Q_{N-2,2k} Q_{N,k}^2,\label{CNqsys}\\
&& Q_{N-1,2k+2}Q_{N-1,2k}=Q_{N-1,2k+1}^2 - Q_{N-2,2k+1} Q_{N,k+1}Q_{N,k} .\nonumber\\
&& Q_{N,k+1}\,Q_{N,k-1}=Q_{N,k}^2-  Q_{N-1,2k},\nonumber\\ \nonumber\\
\g=D_{N}: &&
Q_{a,k+1}Q_{a,k-1}=Q_{a,k}^2-Q_{a+1,k}Q_{a-1,k}, \qquad (a\in [1,N-3]),\nonumber\\
&& Q_{N-2,k+1}Q_{N-2,k-1}=Q_{N-2,k}^2-Q_{N,k}Q_{N-1,k}Q_{N-3,k},\label{DNqsys} \\
&& Q_{a,k+1}Q_{a,k-1}=Q_{a,k}^2-Q_{N-2,k} \qquad (a\in \{N-1,N\}) .\nonumber
\end{eqnarray}
\vskip.2in
These recursion relations (for $k\geq 1$) were originally observed \cite{KR} to be relations satisfied by characters of finite-dimensional irreducible Yangian modules.

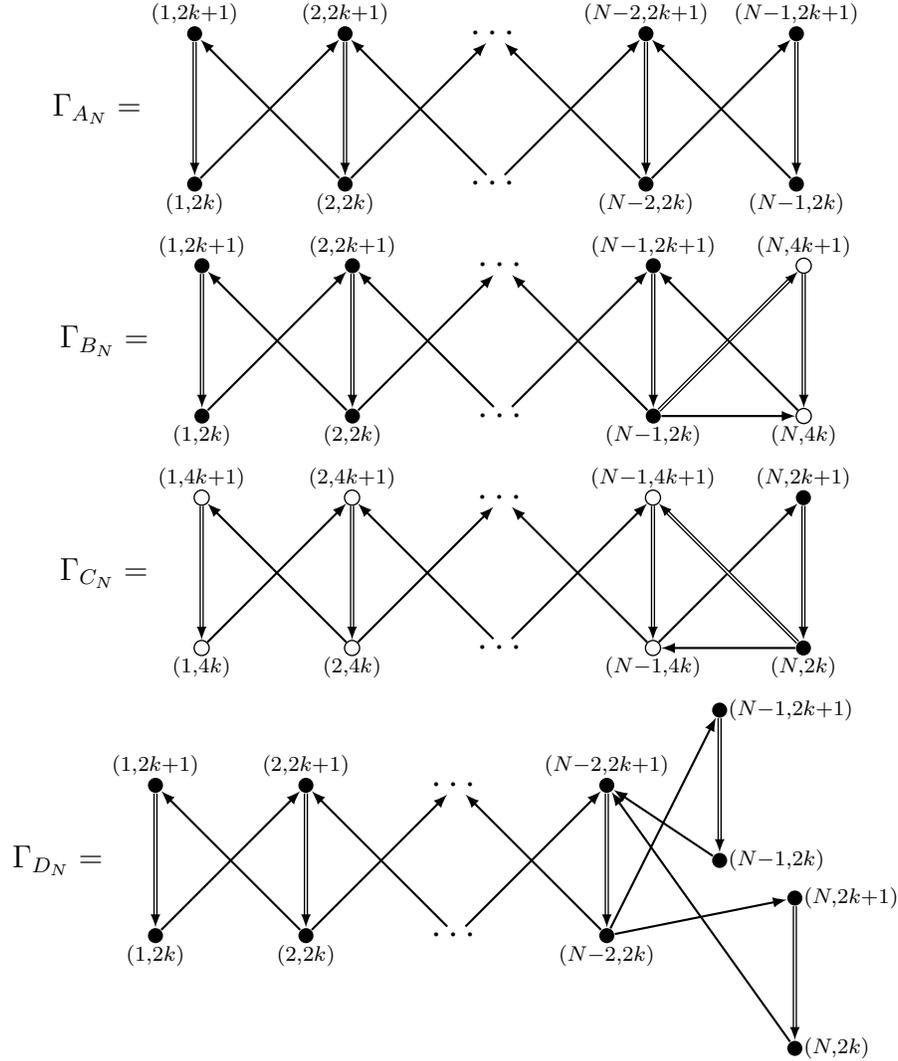
\begin{figure}\centering
\begin{tikzpicture}[scale=1,line width=.5pt, evens/.style={circle,fill=black},odds/.style={shape=circle,color=black,draw},
even/.style={circle,fill=black},odd/.style={shape=circle,draw},inner sep=0pt,minimum size=2mm]
\node[even] (12k) at (0,0) [label=below:${\scriptstyle (1,2k)}$] {};
\node[even] (22k) at (2,0) [label=below:${\scriptstyle (2,2k)}$] {};
\node[even] (12kp) at (0,2) [label=above:${\scriptstyle (1,2k+1)}$] {};
\node[even] (22kp) at (2,2) [label=above:${\scriptstyle (2,2k+1)}$] {};
\node[even] (42kp) at (6,2) [label=above:${\scriptstyle (N-2,2k+1)}$] {};
\node[even] (42k) at (6,0) [label=below:${\scriptstyle (N-2,2k)}$] {};
\node[even] (N2k) at (8,0) [label=below:${\scriptstyle (N-1,2k)}$] {};
\node[even] (N2kp) at (8,2) [label=above:${\scriptstyle (N-1,2k+1)}$] {};
\node (middleeven) at (4,0) {\large$\cdots$};
\node (middleodd) at (4,2) {\large$\cdots$};
\node (label) at (-1.3,1) { $\Gamma_{A_N} = $};
\draw [->,double] (12kp) to (12k);
\draw [->, double] (22kp) to (22k);
\draw [->, double] (42kp) to (42k);
\draw [->, double] (N2kp) to (N2k);
\draw [->,thick] (12k) to (22kp);
\draw [->,thick] (22k) to (12kp);
\draw [->,thick] (22k) to (middleodd);
\draw [->,thick] (middleeven) to (22kp);
\draw [->,thick] (42k) to (N2kp);
\draw [->,thick] (N2k) to (42kp);
\draw[->,thick](42k)to(middleodd);
\draw[->,thick](middleeven)to(42kp);
\end{tikzpicture}
\vskip.1in
\begin{tikzpicture}[scale=1,line width=.5pt, evens/.style={circle,fill=red},odds/.style={shape=circle,color=red,draw},
even/.style={circle,fill=black},odd/.style={shape=circle,draw},
inner sep=0pt,minimum size=2mm]
\node[odd] (N2k) at (8,0) [label=below:${\scriptstyle (N,4k)}$] {};
\node[even] (12k) at (0,0) [label=below:${\scriptstyle (1,2k)}$] {};
\node[even] (22k) at (2,0) [label=below:${\scriptstyle (2,2k)}$] {};
\node[even] (42k) at (6,0) [label=below:${\scriptstyle (N-1,2k)}$] {};
\node[odd] (N2kp) at (8,2) [label=above:${\scriptstyle (N,4k+1)}$] {};
\node[even] (12kp) at (0,2) [label=above:${\scriptstyle (1,2k+1)}$] {};
\node[even] (22kp) at (2,2) [label=above:${\scriptstyle (2,2k+1)}$] {};
\node[even] (42kp) at (6,2) [label=above:${\scriptstyle (N-1,2k+1)}$] {};
\node (middleeven) at (4,0) {\large$\cdots$};
\node (middleodd) at (4,2) {\large$\cdots$};
\node (label) at (-1.3,1) {$\Gamma_{B_N} = $};
\draw [->,double] (12kp) to (12k);
\draw [->, double] (22kp) to (22k);
\draw [->, double] (42kp) to (42k);
\draw [->, double] (N2kp) to (N2k);
\draw [->,thick] (12k) to (22kp);
\draw [->,thick] (22k) to (12kp);
\draw [->,thick] (22k) to (middleodd);
\draw [->,thick] (middleeven) to (22kp);
\draw [->,double] (42k) to (N2kp);
\draw [->,thick] (42k) to(N2k);
\draw [->,thick] (N2k) to (42kp);
\draw[->,thick](42k)to(middleodd);
\draw[->,thick](middleeven)to(42kp);
\end{tikzpicture}\vskip.1in

\begin{tikzpicture}[scale=1,line width=.5pt, evens/.style={circle,fill=red},odds/.style={shape=circle,color=red,draw},
even/.style={circle,fill=black},odd/.style={shape=circle,draw},
inner sep=0pt,minimum size=2mm]
\node[even] (N2k) at (8,0) [label=below:${\scriptstyle (N,2k)}$] {};
\node[odd] (12k) at (0,0) [label=below:${\scriptstyle (1,4k)}$] {};
\node[odd] (22k) at (2,0) [label=below:${\scriptstyle (2,4k)}$] {};
\node[odd] (42k) at (6,0) [label=below:${\scriptstyle (N-1,4k)}$] {};
\node[even] (N2kp) at (8,2) [label=above:${\scriptstyle (N,2k+1)}$] {};
\node[odd] (12kp) at (0,2) [label=above:${\scriptstyle (1,4k+1)}$] {};
\node[odd] (22kp) at (2,2) [label=above:${\scriptstyle (2,4k+1)}$] {};
\node[odd] (42kp) at (6,2) [label=above:${\scriptstyle (N-1,4k+1)}$] {};
\node (middleeven) at (4,0) {\large$\cdots$};
\node (middleodd) at (4,2) {\large$\cdots$};
\node (label) at (-1.3,1) {$\Gamma_{C_N} = $};
\draw [->,double] (12kp) to (12k);
\draw [->, double] (22kp) to (22k);
\draw [->, double] (42kp) to (42k);
\draw [->, double] (N2kp) to (N2k);
\draw [->,thick] (12k) to (22kp);
\draw [->,thick] (22k) to (12kp);
\draw [->,thick] (22k) to (middleodd);
\draw [->,thick] (middleeven) to (22kp);
\draw [->,thick] (42k) to (N2kp);
\draw [->,thick] (N2k) to (42k);
\draw [->,double] (N2k) to (42kp);
\draw[->,thick](42k)to(middleodd);
\draw[->,thick](middleeven)to(42kp);
\end{tikzpicture}\vskip.1in
\begin{tikzpicture}[scale=1,line width=.5pt, evens/.style={circle,fill=red},odds/.style={shape=circle,color=red,draw},
even/.style={circle,fill=black},odd/.style={shape=circle,draw},
inner sep=0pt,minimum size=2mm]
\node[even] (N2k) at (8.5,-1.5) [label=right:${\scriptstyle (N,2k)}$] {};
\node[even] (Nm2k) at (7.5,1) [label=right:${\scriptstyle (N-1,2k)}$] {};
\node[even] (12k) at (0,0) [label=below:${\scriptstyle (1,2k)}$] {};
\node[even] (22k) at (2,0) [label=below:${\scriptstyle (2,2k)}$] {};
\node[even] (42k) at (6,0) [label=below:${\scriptstyle (N-2,2k)}$] {};
\node[even] (Nm2kp) at (7.5,3) [label=right:${\scriptstyle (N-1,2k+1)}$] {};
\node[even] (N2kp) at (8.5,.5) [label=right:${\scriptstyle (N,2k+1)}$] {};
\node[even] (12kp) at (0,2) [label=above:${\scriptstyle (1,2k+1)}$] {};
\node[even] (22kp) at (2,2) [label=above:${\scriptstyle (2,2k+1)}$] {};
\node[even] (42kp) at (6,2) [label=above:${\scriptstyle (N-2,2k+1)}$] {};
\node (middleeven) at (4,0) {\large$\cdots$};
\node (middleodd) at (4,2) {\large$\cdots$};
\node (label) at (-1.3,1) {$\Gamma_{D_N} = $};
\draw[->,thick](42k)to(Nm2kp);
\draw[->,thick](Nm2k)to(42kp);
\draw[->,double](Nm2kp)to(Nm2k);
\draw [->,double] (12kp) to (12k);
\draw [->, double] (22kp) to (22k);
\draw [->, double] (42kp) to (42k);
\draw [->, double] (N2kp) to (N2k);
\draw [->,thick] (12k) to (22kp);
\draw [->,thick] (22k) to (12kp);
\draw [->,thick] (22k) to (middleodd);
\draw [->,thick] (middleeven) to (22kp);
\draw [->,thick] (42k) to (N2kp);
\draw [->,thick] (N2k) to (42kp);
\draw[->,thick](42k)to(middleodd);
\draw[->,thick](middleeven)to(42kp);
\end{tikzpicture}

\caption{\small The quivers for the $A_{N-1},D_N,B_N,C_N$ Q-system cluster algebras. 
We have indicated a generic Q-system cluster
along the bipartite belt: each vertex labelled $(a,k)$ corresponds to a cluster variable $Q_{a,k}$. Nodes corresponding to short roots are denoted by empty circles.}
\label{fig:CNquiv}
\end{figure}

Each set of the Q-systems is associated with a cluster algebra:
\begin{thm}\cite{Ke07,DFKQcluster}
For each algebra $\g$, the variables $\{Q_{a,k}: a\in[1,r], k\in \Z\}$, up to a simple rescaling which eliminates the minus sign on the right hand side, are cluster variables in a corresponding cluster algebra. Each of the Q-system relations is an exchange relation in the cluster algebra.
\end{thm}

The cluster algebras are defined via a $2r\times 2r$ skew symmetric exchange matrix $B$,
($r$ being the rank of $\g$), or quiver $\Gamma$ as in Figure \ref{fig:CNquiv}, which depends only on the Cartan matrix $C$ of $\g$:
\begin{equation}\label{Bmatrix}
B = \left(\begin{array}{c | c } C^T - C & -C^T \\ \hline C & 0\end{array}
\right). 
\end{equation}
This exchange matrix, together with the initial cluster variables $\bX = (Q_{a,0}; Q_{a,1})_{a=1}^r$, defines the cluster algebra. The cluster variables $\{Q_{a,k}\}$ are obtained from a generalized bipartite evolution of the initial cluster $(\bX,B)$ \cite{DFKQcluster}. 
The subset of mutations on the (generalized) bipartite belt which generates all the cluster variables corresponding to the Q-system algebra was given in \cite{DFKQcluster}, Theorem 3.6.

\subsection{Quantum Q-systems}

One of the advantages of formulating the Q-system relations in terms of cluster algebra mutations is that there is a canonical quantization of the cluster algebra, using the canonical Poisson structure \cite{GSV} and its quantization \cite{BZ}.

The quantum cluster algebra attached to a non-degenerate, skew-symmetric matrix $B$ is the non-commutative algebra generated by the cluster variables $\bX=(X_i)$ at an initial cluster and their inverses, with exchange matrix $B$, as well as
the cluster variables at all mutation equivalent clusters.
Within the cluster $(\mathbf X, B)$ the cluster variables $q$-commute according to a skew-symmetric $\Lambda$ proportional to the inverse of $B$:
$$ X_i X_j = q^{\Lambda_{ij}} X_j X_i.$$ 
The exchange relations are given by normal-ordering of the classical mutations:
$$
X_i' = : \prod_{j: B_{ij}>0} X_j^{B_{ij}} X_i^{-1}: + :\prod_{j: B_{ij}<0} X_j^{-B_{ij}} X_i^{-1}: .
$$
Here, given a monomial $\prod X_i^{b_i}$ such that $X_iX_j =q^{a_{i,j}}X_jX_i$, the normal ordered product is
$$:X_1^{b_1} \cdots X_\ell^{b_\ell} :\, = q^{-\frac{1}{2}\sum_{i<j} a_{i,j} b_i b_j}\, X_1^{b_1} \cdots X_\ell^{b_\ell}.$$

The exchange matrices \eqref{Bmatrix} corresponding to the Q-system cluster algebras are skew-symmetric and invertible. We use the associated quantum cluster algebra to define 
the quantum Q-systems as the quantized exchange relations corresponding to the exchange relations appearing in the classical Q-systems.

Let $\Lambda$ be a $2r\times 2r$ skew-symmetric matrix,
with $\Lambda^T B = -D$, a diagonal integer matrix with negative integer entries. Then
\begin{equation}\label{LambdaMatrix}
\Lambda = \left(\begin{array}{c | c } 0 & \lambda \\ \hline -\lambda^T & \lambda^T- \lambda \end{array}
\right),
\end{equation}
where $\lambda$ is proportional to the inverse Cartan matrix.
We use the following normalizations:
\begin{equation}\label{lambda}
\begin{array}{rl}
A_{N-1}: & \lambda_{ab}=C^{-1}_{ab}
=  {\min}(a,b) - \frac{ab}{N} \\ \\
B_N: & \lambda_{ab}= 2 C^{-1}_{ab} =
\left\{ \begin{array}{ll} 2 {\min}(a,b), & 1\leq a\leq N,\ 1\leq b\leq N-1;\\ 
{\min}(a,b), & 1\leq a\leq N=b ,
\end{array}
\right.  
\\ \\
C_N: & \lambda_{ab}= 2 C^{-1}_{ab} =
\left\{ \begin{array}{ll} 2 {\min}(a,b), & 1\leq a\leq N-1,\ 1\leq b\leq N;\\
{\min}(a,b), & 1\leq b\leq N=a,
\end{array}
\right.
\\ \\
D_N: & \lambda_{ab}=2 C^{-1}_{ab} 
=\left\{ \begin{array}{ll}
 2 {\min}(a,b), &  1\leq a,b\leq N-2,\\
b, &  1\leq b\leq N-2,\  a\in \{N-1,N\},\\
a, &  1\leq a\leq N-2,\  b \in \{N-1,N\},\\
\half N, & a=b \in \{N-1,N\},\\
\half (N-2), & a\neq b \in \{N-1,N\}.
\end{array} 
\right.
\end{array}
\end{equation}
This choice of normalization results, in the cases $B_N,C_N,D_N$, in the  value 
$\lambda_{a,b}=2\, {\rm min}(a,b)$ whenever $a,b$ correspond
to the part of the Dynkin diagram that forms an A-type chain, that is, $a,b \in [1,N-1]$ for types $BC$ and $a,b\in [1,N-2]$ for type $D$.

The quantum Q-systems are recursion relations for the non-commuting variables $\Q_{a,k}$. These take the form
\begin{eqnarray}\label{ANqqsys}\!\!\!\!\!\!\!\!\!\!\!\!\!\!\!\!\!\!\!\!\!\!\!\!\!\!\!\!\!\!\!\!\!\!\!\!\!\!\!\!A_{N-1}&:&
\Q_{a,k}\, \Q_{b,p}=q^{\lambda_{a,b}(p-k)} \, \Q_{b,p}\,\Q_{a,k} ,\nonumber \\
&&q^{\lambda_{a,a}} \, \Q_{a,k+1}\,\Q_{a,k-1}= \Q_{a,k}^2-q^{\frac{1}{2}} \,\Q_{a+1,k}\,\Q_{a-1,k},\quad (a\in [1,N]) ,
 \\
\nonumber&&\Q_{0,k}=1, \qquad \Q_{N+1,k}=0.
\end{eqnarray}
\begin{eqnarray}\quad\quad\,  B_N&:&
\Q_{a,k}\, \Q_{b,p}= q^{p \lambda_{a,b}-k \lambda_{b,a}}\, \Q_{b,p}\, \Q_{a,k},\nonumber \\
&&q^{2a} \,\Q_{a,k+1}\,\Q_{a,k-1}=\Q_{a,k}^2-q\, \Q_{a+1,k}\,\Q_{a-1,k}\, \Q_{\al-1,k},\quad (\al\in [1,N-2]),\nonumber \\
&&q^{2N-2} \,\Q_{N-1,k+1}\,\Q_{N-1,k-1}=(\Q_{N-1,k})^2-q \Q_{N,2k}\, \Q_{N-2,k}, \label{BNqqsys}\\
&&q^N\, \Q_{N,2k+1}\,\Q_{N,2k-1}= (\Q_{N,2k})^2 -q (\Q_{N-1,k})^2,\nonumber \\
&&q^N\, \Q_{N,2k+2}\,\Q_{N,2k}= (\Q_{N,2k+1})^2 -q^{N} \Q_{N-1,k+1}\,\Q_{N-1,n}, \nonumber \\
&&\Q_{0,k}=1. \nonumber
\end{eqnarray}

\begin{eqnarray}\quad\quad C_N&:&
\Q_{a,k}\, \Q_{b,p}=
q^{p \lambda_{a,b} -k \lambda_{b,a}}\, 
\Q_{b,p}\, \Q_{a,k},\nonumber \\
&&q^{2a} \Q_{a,k+1}\,\Q_{a,k-1}=\Q_{a,k}^2-q \,\Q_{a+1,k}\,\Q_{a-1,k} ,\qquad (\al\in[1,N-2]),\nonumber\\
&&q^{2N-2} \Q_{N-1,2k+1}\Q_{N-1,2k-1}=\Q_{N-1,2k}^2 -q  \Q_{N-2,2k} \Q_{N,k}^2, \label{CNqqsys}\\
&&q^{2N-2} \Q_{N-1,2k+2}\Q_{N-1,2k}=\Q_{N-1,2k+1}^2 - q^{1+\frac{N}{2}}\,\Q_{N-2,2k+1} \Q_{N,k+1}\Q_{N,k}, \nonumber\\
&&q^N \Q_{N,k+1}\,\Q_{N,k-1}= \Q_{N,k}^2- q \Q_{N-1,2k},\nonumber \\
&&\Q_{0,k}=1. \nonumber
\end{eqnarray}
\begin{eqnarray}
D_N&:&
\Q_{a,k}\, \Q_{b,p}=q^{\lambda_{a,b}(p-k)}\, \Q_{b,p}\, \Q_{a,k},\nonumber \\
&&q^{2a}\Q_{a,k+1}\,\Q_{a,k-1}=\Q_{a,k}^2-q\, \Q_{a+1,k}\,\Q_{a-1,k}\qquad (a\in [1,N-3]),\nonumber \\
&&q^{2(N-2)}\Q_{N-2,k+1}\,\Q_{N-2,k-1}= \Q_{N-2,k}^2-q\,  \Q_{N,k}\,\Q_{N-1,k}\,\Q_{N-3,k} ,\\
&&q^{\frac{N}{2}}\Q_{N-1,k+1}\,\Q_{N-1,k-1}= \Q_{N-1,k}^2-q\,  \Q_{N-2,k},\nonumber \\
&&q^{\frac{N}{2}}\Q_{N,k+1}\,\Q_{N,k-1}= \Q_{N,k}^2-q\,  \Q_{N-2,k} ,\nonumber \\
&&\Q_{0,k}=1. \nonumber\label{DNqqsys}
\end{eqnarray}

In each case, the $q$-commutation relation, i.e. the first equation in each set, holds only for variables within the same cluster, hence the restriction on possible values of the second index. For example, in the case $A_{N-1}$, the restriction is $|p-k|\leq |a-b|+1$. In general, the clusters consisting of Q-system solutions only are parameterized by generalized Motzkin paths \cite{Simon}.

\begin{remark}\label{boundary_condition}
The  general boundary condition in \eqref{ANqqsys} the quantum Q-system of type $A_{N-1}$ is $\Q_{N+1,k}=0$, as opposed to the condition $Q_{N,k}=1$ for the classical system \eqref{ANqsys}, which is compatible but more restrictive.
This means that there is an additional set of  variables in the center of the quantum cluster algebra which satisfy
$\Q_{N,k+1}\Q_{N,k-1}=\Q_{N,k}^2$. This is consistent with the extension of the definition \eqref{lambda} of the matrix $\lambda_{a,b}$ in type $A$ to an $N\times N$ matrix using the same formula, so that $\lambda_{a,N}=0$.
The most general solution subject to the boundary conditions has $\Q_{N,k}=\Q_{N,1}^k \Q_{N,0}^{k-1}$, in terms of the two central elements $\Q_{N,1}$ and $\Q_{N,0}$. This more general system can be embedded into a cluster algebra with coefficients.
\end{remark}

\section{The $A_{N-1}$ case solution: generalized Macdonald difference 
operators and quantum determinants}\label{sec:A}

\subsection{Renormalized quantum Q-system}

With the choice of boundary condition for the type $A$ quantum Q-system as in Remark \ref{boundary_condition}, the quantum Q-system is homogeneous with respect to the grading $\deg(\Q_{a,k}):=a k$ ($a\in[1,N]$). 
We adjoin an invertible degree operator $\Delta^{1/N}$ to the algebra, such that
$$\Delta\, \Q_{a,k}=q^{ak}\, \Q_{a,k}\,\Delta, \qquad a\in [1,N],k\in \Z .$$
Using $\lambda_{a,a}+\frac{a}{N}=a$ and $\lambda_{a+1,a+1}+\lambda_{a-1,a-1}-2\lambda_{a,a}=-\frac{2}{N}$,
the renormalized variables
\begin{equation}\label{renorm} \widetilde\Q_{a,k}=q^{-\frac{1}{2}(k+\frac{N}{2})\lambda_{a,a}}\, \Q_{a,k} \, \Delta^\frac{a}{N}  ,\qquad a\in [1,N], k\in \Z ,
\end{equation}
satisfy the
{\it renormalized quantum Q-system}:
\begin{eqnarray}
&&\widetilde\Q_{a,k}\,\widetilde\Q_{b,k'}= q^{(k'-k){\min}(a,b)} \, \widetilde\Q_{b,k'}\,\widetilde\Q_{a,k},\quad |k-k'| \leq |a-b|+1,\nonumber \\
&&q^a\, \widetilde\Q_{a,k+1}\, \widetilde\Q_{a,k-1} =\widetilde\Q_{a,k}^2- \widetilde\Q_{a+1,k}\, \widetilde\Q_{a-1,k} ,\qquad (a\in [1,N]),\label{Msys}\\
&&\widetilde\Q_{0,k}=1,\qquad \widetilde\Q_{N+1,k}=0.\nonumber
\end{eqnarray}

The choice $\Q_{N,0}=1$ implies $\widetilde\Q_{N,0}=\Delta$.
Defining $A=\widetilde\Q_{N,1}\widetilde\Q_{N,0}^{-1}$, a homogeneous element of degree N, we have $\widetilde\Q_{N,k}=A^k\, \Delta$. The quiver with coefficients corresponding to this cluster algebra is illustrated in Figure \ref{qAquiver}.

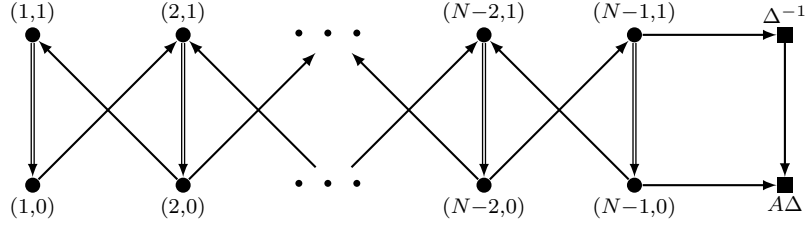
\begin{figure}
\begin{tikzpicture}[scale=1,line width=.5pt, evens/.style={circle,fill=red},odds/.style={shape=circle,color=red,draw},
even/.style={circle,fill=black},odd/.style={shape=circle,draw},
coeff/.style={shape=rectangle,fill=black},
inner sep=1pt,minimum size=2mm]
\node[even] (N2k) at (8,0)  [label=below:${\scriptstyle (N-1,0)}$]{};
\node[even] (12k) at (0,0)  [label=below:${\scriptstyle (1,0)}$] {};
\node[even] (22k) at (2,0)  [label=below:${\scriptstyle (2,0)}$]{};
\node[even] (42k) at (6,0)  [label=below:${\scriptstyle (N-2,0)}$]{};
\node[even] (N2kp) at (8,2) [label=above:${\scriptstyle (N-1,1)}$] {};
\node[even] (12kp) at (0,2)  [label=above:${\scriptstyle (1,1)}$] {};
\node[even] (22kp) at (2,2)  [label=above:${\scriptstyle (2,1)}$] {};
\node[even] (42kp) at (6,2)  [label=above:${\scriptstyle (N-2,1)}$] {};
\node (middleeven) at (4,0) {\Huge$\cdots$};
\node (middleodd) at (4,2) {\Huge$\cdots$};
\node[coeff] (A) at (10,2) [label=above:${\scriptstyle \Delta^{-1}}$] {};
\node[coeff] (D) at (10,0) [label=below:${\scriptstyle A\Delta}$] {};
\draw [->, double] (12kp) to (12k);
\draw [->,double] (22kp) to (22k);
\draw [->,double] (42kp) to (42k);
\draw [->,double] (N2kp) to (N2k);
\draw [->,thick] (12k) to (22kp);
\draw [->,thick] (22k) to (12kp);
\draw [->,thick] (22k) to (middleodd);
\draw [->,thick] (middleeven) to (22kp);
\draw [->,thick] (42k) to (N2kp);
\draw [->,thick] (N2k) to (42kp);
\draw[->,thick](42k)to(middleodd);
\draw[->,thick](middleeven)to(42kp);
\draw[->,thick] (N2k) to (D);
\draw[->,thick] (N2kp) to (A);
\draw[->,thick] (A) to (D);
\end{tikzpicture}
\caption{The type $A$ quantum Q-system quiver corresponding to the initial seed $\{\Q_{a,0},\Q_{a,1}\}$. Square nodes denote coefficients.}\label{qAquiver}
\end{figure}

\subsection{The quantum determinant}
The exchange relations 
\eqref{Msys} define a quantum determinant: The variables $\widetilde\Q_{a,k}$ with $a>1$ are polynomials in
the variables $\{\widetilde\Q_{1,k'}: |k'-k| \leq a-1\}$. Below, we use the notation $\widetilde\Q_{k}:=\widetilde\Q_{1,k}$.
The quantum determinant is best defined in terms of generating functions.
\begin{defn} Given a set of integers $k_1,...,k_a\in \Z$, define the Hankel matrix 
\begin{equation}\label{hank}
(\widetilde\Q_{k_i+i-j})_{1\leq i,j\leq a} .
\end{equation}
The quantum determinant of this matrix,  
denoted by $\widetilde\Q[k_1,k_2,...,k_a]$,
is given by the coefficients of the generating function:
$$\sum_{k_1,...,k_\al\in \Z} u_1^{k_1}\cdots u_a^{k_\al}\, {\widetilde\Q}[k_1,...,k_a]=
\prod_{1\leq i<j \leq a} \left(1-q\frac{u_j}{u_i}\right) { \widetilde\Q}(u_1) {\widetilde\Q}(u_2)\cdots  {\widetilde\Q}(u_\al),$$
where 
$${\widetilde\Q}(u):= \sum_{k\in \Z} u^k\, \widetilde\Q_{k}.$$
\end{defn}
The quantum determinant $\widetilde\Q[k_1,...,k_a]$ is a homogeneous polynomial of degree $a$ in the $\widetilde\Q_{k}$s. 

\begin{thm}{\cite{DFK16}}
The solutions $\widetilde\Q_{a,k}$ of the system \eqref{Msys} with $a\geq 1$ and $k\in \Z$ are the quantum determinants
$$\widetilde\Q_{a,k}= 
\widetilde\Q[\ \underbrace{ k,k,...,k }_{\text{ $a$ times}}\ ] .$$
\end{thm}

\subsection{The functional representation of the quantum Q-system}
We recall the functional representation of the renormalized quantum Q-system \eqref{Msys} $\rho (\widetilde\Q_{a,k})=M_{a,k}$ \cite{DFK15}, which act
on the space of symmetric functions of $N$ variables $x_1,x_2,...,x_N$.
\begin{thm}\cite{DFK15,DFK16}\label{macdoQsys}
Let $\Gamma_i$ be the $q$-shift operator acting on the space of functions in $N$ variables, defined by
$\Gamma_i f(x_1,...,x_i,...,x_N) = f(x_1,...,q x_i, ..., x_N).$ The $q$-difference operators
\begin{equation}\label{macMo} 
M_{a,k}=\sum_{I\subset [1,N]\atop |I|=a} \Big(\prod_{i\in I}x_i\Big)^k
\prod_{i\in I\atop j\not\in I} \frac{x_i}{x_i-x_j}\, \prod_{i\in I} \Gamma_i ,\quad a\in[1,N], k\in\Z
\end{equation}
satisfy the quantum Q-system relations \eqref{Msys}.
\end{thm}
In this representation,
\begin{equation}
A=x_1x_2\cdots x_N,\qquad \Delta=\Gamma_1\Gamma_2\cdots \Gamma_N .
\end{equation}
\begin{remark}\label{glvssl} For $\sl_N$-characters, the products $A$ and $\Delta$ are taken to be equal to $1$. The more general boundary condition here corresponds to $\gl_N$-characters.
\end{remark}

\subsection{The spherical double affine Hecke algebra}
The reader will have recognized that the difference operators $M_{a,0}$ in  \eqref{macMo} are the 
limit $t\to \infty$  of the (renormalized) Macdonald difference operators in type $A$. This is the main observation which led to the results of \cite{DFKqt}, where it is shown that the spherical DAHA \cite{Cheredbook} of type $A_{N-1}$
is the natural $t$-deformation of the quantum Q-system. 

An important observation in \cite{DFKqt} is that the evolution in the discrete time variable $k$ is induced by the adjoint action of one of
the generators of the $SL_2(\Z)$ symmetry of the DAHA.
\begin{thm}{\cite{DFKqt}}\label{gauM}
The discrete time evolution of the operators $M_{a,k}$ is induced by the adjoint action of the Gaussian $\gamma^{-1}$,
with $\gamma$ as in 
\eqref{gaussian}:
\begin{equation}\label{adjA}
M_{a,k}=q^{-a k/2}\, \gamma^{-k}\, M_{a,0}\, \gamma^k .
\end{equation}
\end{thm}

The difference operators $\{M_{a,k}\}$, together with the elementary symmetric functions of $\bx$, are the image in the functional representation of the generators of the spherical DAHA.

\bigskip
The double affine Hecke algebra and  corresponding Macdonald operators are defined for  other Lie algebras.
In the following sections, we will use this as the inspiration to give conjectures for the functional representations of the quantum Q-systems for the other
classical types by following the inverse reasoning: Starting from the appropriate choice of Macdonald difference operators, act with the adjoint action of $\gamma$ to obtain the discrete
time-evolved operators, imitating the contents of Theorem \ref{gauM}. The resulting $q$-difference operators, in the limit $t\to\infty$, are  (conjecturally) solutions
of renormalized quantum Q-systems for the other classical types.

\subsection{Raising and lowering operators}

The dual $q$-Whittaker functions are defined as
$$\Pi_\lambda\equiv\Pi_\lambda(q^{-1};\bx)=\lim_{t\to \infty} P_\lambda (q,t;\bx),$$
where the $P_\lambda$ are the Macdonald polynomials for type $A_{N-1}$ \cite{macdo}.  Here, the partition $\lambda$ corresponds to a dominant  integral weight of $\gl_N$,
 $\lambda=\sum_{i=1}^N \lambda_i e_i$ in the standard basis of $\R^N$,
with 
$\lambda_1\geq \lambda_2 \geq \cdots \geq \lambda_N\geq 0$. 
The polynomials $P_\lambda$ 
are common monic
eigenfunctions of the Macdonald operators $D_a$
\begin{equation}\label{macdopA}
D_a=\sum_{I\subset [1,N]\atop |I|=a} 
\prod_{i\in I\atop j\not\in I} \frac{t x_i-x_j}{x_i-x_j}\, \prod_{i\in I} \Gamma_i ,  \quad a\in [1,N],
\end{equation}
with 
$$D_a P_\lambda = t^{-\frac{a(a-1)}{2}}\, e_a(q^{\lambda_1} t^{N-1}, q^{\lambda_2} t^{N-2},\ldots, q^{\lambda_N} )P_\lambda,$$
where $e_a$ are the elementary symmetric functions in $N$ variables. 
Using $\displaystyle\lim_{t\to\infty} t^{-a(N-a)}\, D_a=M_{a,0}$, this means that the $q$-Whittaker functions are eigenfunctions of the difference operators $M_{a,0}$:
\begin{equation}\label{eigenAN}
M_{a,0}\, \Pi_\lambda= q^{(\lambda,\omega_a)}\, \Pi_\lambda , \quad a\in [1,N],
\end{equation}
where $\omega_a=e_1+e_2+\cdots +e_a$ are the fundamental weights of $\mathfrak{gl}_N$.

In \cite{DFK15}  we showed that the operators $M_{a;\pm 1}$ are the limit $t\to\infty$
of the raising and lowering operators for Macdonald polynomials constructed by Kirillov and Noumi \cite{Kinoum}. However, the commutation relations in \eqref{Msys}  provide an easier proof, which we present here.
\begin{thm}\label{KNAN}
The operators $M_{a,\pm 1}$ act on the polynomials $\Pi_\lambda$ as follows:
\begin{eqnarray}
M_{a,1}\, \Pi_\lambda&=& q^{(\lambda,\omega_a)}\, \Pi_{\lambda+\omega_a}, \label{raisingAN}\\
M_{a,-1}\, \Pi_\lambda&=& q^{(\lambda,\omega_a)}\, \big(1-q^{-(\lambda,\al_{a})}\big)\, \Pi_{\lambda-\omega_a}, \label{loweringAN}
\end{eqnarray}
where $\al_a=e_a-e_{a+1}$
are the simple roots of $\sl_N$. 
\end{thm}
Note that if $\lambda$ is a partition, $\lambda+\omega_a$ is also a partition, and if $\lambda-\omega_a$ is not a partition, the scalar factor $(1-q^{(\lambda,\alpha_a)})$ vanishes.
\begin{proof}
Up to scalar multiple, the polynomials $\Pi_\lambda$ are uniquely determined by the collection of eigenvalues $q^{(\omega_a,\lambda)}$ corresponding to the diagonal action of the operators $M_{a,0}$,
$a=1,2,...,N$.
The commutation relations in \eqref{Msys} are
$$ M_{a,0}\, M_{b,\pm 1} =q^{\pm {\min}(a,b)} \, M_{b,0\pm 1}\, M_{a,0}. 
$$
Applying this to $\Pi_\lambda$ and using \eqref{eigenAN},
$$ M_{a,0}\, M_{b,\pm 1} \Pi_\lambda=q^{(\lambda,\omega_a)\pm  {\min}(a,b)}\, M_{b,\pm 1} \Pi_\lambda=q^{(\lambda\pm\omega_b,\omega_a)} M_{b,\pm1}\Pi_\lambda,$$
since ${\min}(a,b)=(\omega_b,\omega_a)$. 
That is, $M_{b,\pm 1} \Pi_\lambda$ is an eigenfunction of $M_{a,0}$
with eigenvalue $q^{(\lambda\pm \omega_b,\omega_a)}$. Therefore
there exist scalars 
$c^\pm_{\lambda,b}$ such that
$M_{b,\pm 1} \Pi_\lambda =c^\pm_{\lambda,b}\,\Pi_{\lambda\pm \omega_b}$.

Recall that the Macdonald polynomials $P_\lambda$, and therefore $\Pi_\lambda$,  have a triangular decomposition with respect to the monomial symmetric functions with leading term $m_\lambda$. The following analysis provides the scalar factor $c^+_{\lambda,b}$.
When
$|x_1|\gg|x_2|\gg \cdots \gg|x_N|$, $\Pi_\lambda =x_1^{\lambda_1}\cdots x_{N}^{\lambda_N}+{\rm lower \ order}$ and
$M_{b,1}=x_1x_2...x_b\, \Gamma_1\Gamma_2\cdots \Gamma_b+{\rm lower \ order}$. Thus,
$$ M_{b,1}\, \Pi_\lambda= q^{\lambda_1+\cdots +\lambda_b} x_1^{\lambda_1+1}\cdots x_\beta^{\lambda_b+1}
x_{b+1}^{\lambda_{b+1}}
\cdots x_n^{\lambda_N}+{\rm lower \ order}=q^{(\lambda,\omega_b)}\Pi_{\lambda+\omega_b}+{\rm lower \ order},$$
and thus $c^+_{\lambda,b}=q^{(\lambda,\omega_b)}$. 

Applying the exchange relation in \eqref{Msys} with $k=0$
to $\Pi_\lambda$, we get
$$q^a\, M_{a,1}\, M_{a,-1}\, \Pi_\lambda=q^{a} \, c^+_{\lambda-\omega_a,a}\, c^-_{\lambda,a} \Pi_\lambda=
(q^{2(\lambda,\omega_a)}-q^{(\lambda,\omega_{a+1}+\omega_{a-1})}) \Pi_\lambda,$$
which shows that $c^-_{\lambda,a}=q^{(\lambda,\omega_a)}-q^{(\lambda,\omega_{a+1}+\omega_{a-1}-\omega_a)}$ and the Theorem follows.
\end{proof}

\subsection{Graded characters in terms of difference operators}
The difference operators \eqref{macMo} can be used to efficiently generate graded characters \eqref{linfun}. Theorem \ref{macdoQsys} is a necessary condition for the following theorem:
\begin{thm} (\cite{DFK15}, Corollary 18):
Starting with the trivial character $\chi_0=1$, the difference operators \eqref{macMo} act consecutively to generate
the character of the graded tensor product of KR-modules 
$\otimes_{a,i} \KR_{i\omega_a}^{\otimes n_{a,i}}$ as follows: 
\begin{eqnarray}
\chi_\bn(q^{-1},\bx)&= & q^{-\half Q(\bn)}
 \prod_{a=1}^{N-1}(M_{a,k})^{n_{a,k}}\prod_{a=1}^{N-1}(M_{a,k-1})^{n_{a,k-1}}\cdots
\prod_{a=1}^{N-1}(M_{a,1})^{n_{a,1}}\, 1.
\end{eqnarray}
where
$$
Q(\bn) =  \sum_{a,b=1}^{N-1}\sum_{i,j\geq 1} n_{a,i}\, {\rm min}(i,j)\, {\rm min}(a,b)\, n_{b,j}-\sum_{a=1}^{N-1}\sum_{i\geq 1} i\, a\,  n_{a,i}.
$$
\end{thm}

Alternatively, when $k\geq \max\{j: n_{a,j}>0\}$, 
$$
D_{a,k} \chi_{\bn} (q^{-1};\bx) = q^{(\omega_a,\sum_{b,j} j n_{b,j}\omega_b)}\chi_{\bn+\epsilon_{a,k}}(q^{-1};\bx),
$$
where we write $\bn=\sum_{b,j} n_{b,j}\epsilon_{b,j}$.

\section{The quantum Q-system conjectures for types $BCD$} \label{sec:conj1}

In order to formulate the conjectures for the functional representation of the quantum Q-systems of types $D_N,B_N,C_N$, we start with the results
of Macdonald and van Diejen. There are $N$ algebraically independent commuting Hamiltonians in the sDAHA of those types. Our goal is to find a set that will be the seed, in the $q$-Whittaker limit, of the Q-system solutions.

\subsection{Macdonald and van Diejen operators}
\label{diff_ops}

In \cite{macdoroot}, Macdonald constructed certain difference operators for types $D_N,B_N,C_N$, corresponding to minuscule coweights. For these types, there are, respectively,
$3,1$ and $1$ minuscule coweights, indexed by  some of the extremal nodes of the Dynkin diagrams. The corresponding Macdonald operators are
\begin{eqnarray}
D_N:\quad {\mathcal E}_1^{(D_N)}&=&\sum_{\epsilon=\pm 1}\sum_{i=1}^N \prod_{j\neq i} \frac{1-t x_i^\epsilon x_j}{1-x_i^\epsilon x_j}
\frac{t x_i^\epsilon-x_j}{x_i^\epsilon-x_j} \Gamma_i^{2\epsilon} ,\label{macDN1}\\
{\mathcal E}_{N-1}^{(D_N)}&=&  \sum_{\epsilon_1,...,\epsilon_N=\pm 1\atop \epsilon_1\epsilon_2\cdots \epsilon_N=-1}  
\prod_{1\leq i<j\leq N} \frac{1-t x_i^{\epsilon_i} x_j^{\epsilon_j}}{1-x_i^{\epsilon_i}x_j^{\epsilon_j}} 
\prod_{i=1}^N \Gamma_i^{\epsilon_i}, \label{macDN2}\\
{\mathcal E}_N^{(D_N)}&=&  
\sum_{\epsilon_1,...,\epsilon_N=\pm 1\atop \epsilon_1\epsilon_2\cdots \epsilon_N=1}  \prod_{1\leq i<j\leq N} \frac{1-t x_i^{\epsilon_i} x_j^{\epsilon_j}}{1-x_i^{\epsilon_i}x_j^{\epsilon_j}} \prod_{i=1}^N \Gamma_i^{\epsilon_i}.\label{macDN3}\\
B_N:\quad {\mathcal E}_{1}^{(B_N)}&=&\sum_{\epsilon=\pm 1}
\sum_{i=1}^N \frac{1-t x_i^{\epsilon}}{1-x_i^{\epsilon}}
\prod_{j\neq i} \frac{1-t x_i^{\epsilon} x_j}{1-x_i^{\epsilon}x_j}\frac{t x_i^{\epsilon} -x_j}{x_i^{\epsilon}-x_j}
\Gamma_i^{2\epsilon}.\label{macBN}\\
C_N:\quad {\mathcal E}_{N}^{(C_N)}&=&\sum_{\epsilon_1,...,\epsilon_N=\pm 1}\prod_{i=1}^N
\frac{1-t x_i^{2\epsilon_i}}{1-x_i^{2\epsilon_i}} 
 \prod_{1\leq i<j\leq N}  \frac{1-t x_i^{\epsilon_i}x_j^{\epsilon_j}}{1-x_i^{\epsilon_i}x_j^{\epsilon_j}}
\prod_{i=1}^N \Gamma_i^{\epsilon_i}.\label{macCN}
\end{eqnarray}

The Macdonald operator for $C_N$ is of ``order" $N$, namely acts by shifts of  $N$ variables in each term, as opposed to the operators ${\mathcal E}_1^{(D_N)}$ and 
${\mathcal E}_{1}^{(B_N)}$, which are linear combinations of shifts of a single variable.
A first order difference operator for $C_N$, $ {\mathcal E}_{1}^{(C_N)}$, can be obtained using the commuting operators constructed in \cite{vandiej,vandiejtwo}. 
We choose the following first order $C_N$ operator, which is a particular linear combination of the identity and the first order van Diejen operator:

\begin{equation}
{\mathcal E}_{1}^{(C_N)}=(1+t^{N+1})\frac{1-t^{N}}{1-t}+ \sum_{\epsilon=\pm 1}\sum_{i=1}^N \frac{1-t x_i^{2\epsilon}}{1-x_i^{2\epsilon}} 
\frac{1-t q^2 x_i^{2\epsilon}}{1-q^2 x_i^{2\epsilon}} \prod_{j\neq i}  \frac{1-t x_i^{\epsilon}x_j}{1-x_i^{\epsilon}x_j}
 \frac{t x_i^{\epsilon}-x_j}{x_i^{\epsilon}-x_j} \, (\Gamma_i^{2\epsilon}-1) \label{macCN1}
\end{equation}

This particular choice ensures that the corresponding eigenfunctions are the Macdonald polynomials $P_\lambda(\bx)$ of type $C_N$ with eigenvalues 
$$t^N\, e_1(\{ t^{N+1-i} q^{2\lambda_i}\}_{i=1}^N),\qquad \hbox{where }e_1(z_1,z_2,...,z_N)=\sum_{i=1}^N (z_i+z_i^{-1}).$$
In particular, for $P_0(\bx)=1$, we have ${\mathcal E}_{1}^{(C_N)} 1= t^N\,e_1(\{t,t^2,...,t^{N}\})=(1+t^{N+1})\frac{1-t^{N}}{1-t}$. 

\subsection{The $q$-Whittaker limit}\label{Moperators}

Each of the operators $\E_i^{(\g)}$ in Equations \eqref{macDN1}--\eqref{macCN1} have suitable limits as $t\to\infty$, denoted by
$$M_{i,0}^{(\g)}:=\lim_{t\to\infty} t^{-a_i^{(\g)}} {\mathcal E}_i^{(\g)},$$
where:
\begin{eqnarray*}
a_1^{(D_N)}&=&2(N-1),\quad a_{N-1}^{(D_N)}=a_{N}^{(D_N)}=\frac{N(N-1)}{2},\\
a_1^{(B_N)}&=&2N-1,\\ 
a_1^{(C_N)}&=&2N,\quad a_N^{(C_N)}=\frac{N(N+1)}{2} .
\end{eqnarray*}

It is also necessary to define the additional operator in type $C_N$: 
\begin{equation}\label{macCN2}
M_{1,1}^{(C_N)}:= \sum_{\epsilon=\pm 1}\sum_{i=1}^N \frac{x_i^{2\epsilon}}{x_i^{2\epsilon}-1} 
\frac{q^2 x_i^{2\epsilon}}{q^2 x_i^{2\epsilon}-1} \prod_{j\neq i}  \frac{x_i^{\epsilon}x_j}{x_i^{\epsilon}x_j-1}
 \frac{x_i^{\epsilon}}{x_i^{\epsilon}-x_j} \, x_i^{-\epsilon}(x_i^{2\epsilon} \Gamma_i^{2\epsilon}-q^{-2})
\end{equation}

\subsection{Discrete time evolution by adjoint action of the Gaussian}

By analogy with the $A_{N-1}$ case (see Theorem \ref{gauM}), define the Gaussian function for types $BCD$:
\begin{equation} \gamma:= e^{\sum_{i=1}^N \frac{({\rm Log}(x_i))^2}{4{\rm Log}(q)}}
\end{equation}\label{gau}
{Note the slight modification $q\to q^2$ compared to the $A$ type Gaussian \eqref{gaussian}.} 

The adjoint action of Gaussian function on the difference operators of Section \ref{Moperators} induces a discrete time evolution:
\begin{defn}\label{discrete_time_evolution}
Let $k\in \Z$. The difference operators $M_{a,k}^{(\g)}$ are defined as follows:
\begin{eqnarray*}
D_N:\quad {M}_{1,k}^{(D_N)}&:=&q^{-k}\, \gamma^{-k}\, M_{1,0}^{(D_N)} \, \gamma^k \\
{M}_{N-1,k}^{(D_N)}&:=&q^{-kN/4}\, \gamma^{-k}\, {M}_{N-1,0}^{(D_N)} \, \gamma^k\\
{M}_{N,k}^{(D_N)}&:=&q^{-kN/4}\, \gamma^{-k}\, {M}_{N,0}^{(D_N)} \, \gamma^k\\
B_N:\quad  {M}_{1,k}^{(B_N)}&:=&q^{-k}\, \gamma^{-k}\, {M}_{1,0}^{(B_N)} \, \gamma^k \\
C_N:\quad {M}_{1,2k}^{(C_N)}&:=&q^{-2k}\, \gamma^{-2k}\,{M}_{1,0}^{(C_N)}\,\gamma^{2k}\\
{M}_{1,2k+1}^{(C_N)}&:=&q^{-2k}\, \gamma^{-2k}\,M_{1,1}^{(C_N)}\,\gamma^{2k}\\
{M}_{N,k}^{(C_N)}&:=& q^{-kN/2}\, \gamma^{-2k}\, {M}_{N,0}^{(C_N)}\, \gamma^{2k }
\end{eqnarray*}
\end{defn}
Note that in type $C$, the 
definition for the time evolution of $M_{1,k}^{(C_N)}$ splits into two separate evolutions for even and odd $k$,
involving the operators $M_{1,0}^{(C_N)}$ for even $k$ and $M_{1,1}^{(C_N)}$ of \eqref{macCN2} for odd $k$.
This is due to the fact that $\al_1$ is a short root in type $C$. The same parity phenomenon is expected for the
short root $\al_N$ of $B_N$.

Using the simple relation $\gamma^{-1}\,\Gamma_i^2\, \gamma= q x_i \, \Gamma_i^2$ for $\gamma$ as in \eqref{gau}, we see that Definition \ref{discrete_time_evolution} results immediately in the following explicit expressions: 

\begin{eqnarray}\hbox{Type $D_N$:}\qquad
{M}_{1,k}^{(D_N)}&=&\sum_{\epsilon=\pm 1}\sum_{i=1}^N \prod_{j\neq i} \frac{x_i^\epsilon x_j}{x_i^\epsilon x_j-1}
\frac{x_i^\epsilon}{x_i^\epsilon-x_j} x_i^{k\epsilon}\, \Gamma_i^{2\epsilon}\nonumber\\
{M}_{N-1,k}^{(D_N)}&=&\sum_{\epsilon_1,...,\epsilon_N=\pm 1\atop \epsilon_1\epsilon_2\cdots \epsilon_N=-1}  
\prod_{1\leq i<j\leq N} \frac{x_i^{\epsilon_i} x_j^{\epsilon_j}}{x_i^{\epsilon_i}x_j^{\epsilon_j}-1} 
\prod_{i=1}^N x_i^{\frac{k\epsilon_i}{2}}\, \prod_{i=1}^N \Gamma_i^{\epsilon_i}\label{Dtwo}\\
{M}_{N,k}^{(D_N)}&=&\sum_{\epsilon_1,...,\epsilon_N=\pm 1\atop \epsilon_1\epsilon_2\cdots \epsilon_N=1}  
\prod_{1\leq i<j\leq N} \frac{x_i^{\epsilon_i} x_j^{\epsilon_j}}{x_i^{\epsilon_i}x_j^{\epsilon_j}-1} 
\prod_{i=1}^N x_i^{\frac{k\epsilon_i}{2}}\, \prod_{i=1}^N \Gamma_i^{\epsilon_i}\label{Dthree}
\end{eqnarray}
\begin{equation*}
\hbox{Type $B_N:$}\qquad {M}_{1,k}^{(B_N)}=\sum_{\epsilon=\pm 1}
\sum_{i=1}^N \frac{x_i^{\epsilon}}{x_i^{\epsilon}-1}
\prod_{j\neq i} \frac{x_i^{\epsilon} x_j}{x_i^{\epsilon}x_j-1}\frac{x_i^{\epsilon} }{x_i^{\epsilon}-x_j}
x_i^{k\epsilon}\,\Gamma_i^{2\epsilon}
\end{equation*}
\begin{equation*}
\hbox{Type $C_N$:}\qquad{M}_{1,2k}^{(C_N)}=q^{-2k}+ \sum_{\epsilon=\pm 1}\sum_{i=1}^N \frac{x_i^{2\epsilon}}{x_i^{2\epsilon}-1} 
\frac{q^2 x_i^{2\epsilon}}{q^2 x_i^{2\epsilon}-1} \prod_{j\neq i}  \frac{x_i^{\epsilon}x_j}{x_i^{\epsilon}x_j-1}
 \frac{x_i^{\epsilon}}{x_i^{\epsilon}-x_j} \, (x_i^{2k\epsilon}\Gamma_i^{2\epsilon}-q^{-2k})
 \end{equation*}
 \begin{eqnarray}
{M}_{1,2k-1}^{(C_N)}&=&\sum_{\epsilon=\pm 1}\sum_{i=1}^N \frac{x_i^{2\epsilon}}{x_i^{2\epsilon}-1} 
\frac{q^2 x_i^{2\epsilon}}{q^2 x_i^{2\epsilon}-1} \prod_{j\neq i}  \frac{x_i^{\epsilon}x_j}{x_i^{\epsilon}x_j-1}
 \frac{x_i^{\epsilon}}{x_i^{\epsilon}-x_j} \, x_i^{-\epsilon}(x_i^{2k\epsilon} \Gamma_i^{2\epsilon}-q^{-2k})\nonumber \\
{M}_{N,k}^{(C_N)}&=&\sum_{\epsilon_1,...,\epsilon_N=\pm 1}\prod_{i=1}^N
\frac{x_i^{2\epsilon_i}}{x_i^{2\epsilon_i}-1} 
 \prod_{1\leq i<j\leq N}  \frac{x_i^{\epsilon_i}x_j^{\epsilon_j}}{x_i^{\epsilon_i}x_j^{\epsilon_j}-1}
\prod_{i=1}^N x_i^{k\epsilon_i}\,\prod_{i=1}^N \Gamma_i^{\epsilon_i}.\label{Cthree}
\end{eqnarray}

These expressions motivate the main conjectures of this paper. 

\subsection{Other Macdonald operators via quantum determinants}

The list of operators defined in the previous section  can be completed to include $M_{a,k}^{(\g)}$ for the Dynkin labels $a$ which are not in the list above, by using quantum determinants.

Comparing the 
commutation relations and mutation relations of the quantum $BCD$ Q-systems involving only the Dynkin labels $\{1,...,N-1\}$ (types $BC$) and $\{1,...,N-2\}$ (type $D$), we see that the relations are the same
as in type $A$ of \eqref{Msys}, expressed in the variable $M_{a,k}$, up to the change
of parameter $q\mapsto q^2$ (which results in a change in the first factor $q^a \mapsto q^{2a}$, and the commutation relations $q^{\min(a,b)}
\mapsto q^{2{\min(a,b)}}$). This observation leads to the following definitions:

\begin{defn}
For types $BCD$, define the difference operators $M_{a,k}^{(\g)}$
by the following quantum determinants:
\begin{equation}\label{firmac}
M^{(\g)}_{a,k}:= |{\bf M}^{(\g)}(\{\underbrace{k,k,...,k}_{a\ \rm{times}}\})|_{q^2}=M_{k,k,...,k}^{(\g)} \qquad (a=1,2,...,n_{\g})
\end{equation}
where $n_{D_N}=N-2$, and $n_{B_N}=n_{C_N}=N-1$. Here the matrix ${\bf M}^{(\g)}(\{\ba \})$
is obtained by replacing $M_{1,k}$ with $M^{(\g)}_{1,k}$ in the defining expression \eqref{hank},
and the quantum determinant has parameter $q^2$ instead of $q$ in the vandermonde factor.
\end{defn}

By homogeneity of the quantum determinant as a polynomial in the variables $M^{(\g)}_{1,k}$,  Definition \eqref{firmac} is
compatible with the discrete time evolution relations:
\begin{eqnarray*}
M_{a,k}^{(D_N)}&=&q^{-k a}\, \gamma^{-k} \, M_{a,0}^{(D_N)} \, \gamma^k,\qquad (a=1,2,...,N-2);\\
M_{a,k}^{(B_N)}&=&q^{-k a}\, \gamma^{-k} \, M_{a,0}^{(B_N)} \, \gamma^k,\qquad (a=1,2,...,N-1);\\
M_{a,2k+\eta}^{(C_N)}&=&q^{-2k a}\, \gamma^{-2k} \, M_{a,\eta}^{(C_N)} \, \gamma^{2k},\qquad (a=1,2,...,N-1,\eta=0,1),
\end{eqnarray*}
with the Gaussian function $\gamma$ as in \eqref{gau}.

\subsection{The quantum Q-system conjectures}

The main conjecture of this paper is the following:

\begin{conj}\label{qqconj}
For $G=D_N,B_N,C_N$, the Macdonald operators $M_{a,n}^{(\g)}$ obey
the following renormalized versions (we omit the superscript $\g$
for simplicity) of the quantum Q-systems (\ref{DNqqsys}-\ref{CNqqsys}):
\begin{eqnarray}
\hbox{Type $D_N$:}\qquad\qquad\qquad M_{a,k}\, M_{b,p}&=&q^{\Lambda_{a,b}(p-k)}\, M_{b,p}\, M_{a,k}\nonumber \\
q^{2 a}M_{a,k+1}M_{a,k-1}&=&M_{a,k}^2-M_{a+1,k}M_{a-1,k}\qquad (a\in[1,N-3])\nonumber \\
q^{2(N-2)}M_{N-2,k+1}M_{N-2,k-1}&=& M_{N-2,k}^2- q^{-\frac{(N-2)k}{2}}\, M_{N,k}M_{N-1,k}M_{N-3,k}
\label{DMsys} \\
q^{\frac{N}{2}}M_{N-1,k+1}M_{N-1,k-1}&=& M_{N-1,k}^2- q^{\frac{(N-4)k}{2}}\,M_{N-2,k}\nonumber \\
q^{\frac{N}{2}}M_{N,k+1}M_{N,k-1}&=& M_{N,k}^2- q^{\frac{(N-4)k}{2}}\,M_{N-2,k}\nonumber
\end{eqnarray}
\begin{eqnarray}
\hbox{Type $B_N$}:\qquad \qquad\qquad M_{a,k}\, M_{b,p}&=& q^{\Lambda_{a,b} p-\Lambda_{b,a}k}\, M_{b,p}\, M_{a,k}\nonumber \\
q^{2a} \,M_{a,k+1}\,M_{a,k-1}&=&(M_{a,k})^2-M_{a+1,k}\, M_{a-1,k}\quad (a\in [1,N-2])\nonumber \\
q^{2N-2} \,M_{N-1,k+1}\,M_{N-1,k-1}&=&(M_{N-1,k})^2-M_{N,2k}\, M_{N-2,k}\label{BMsysthird} \\
q^N\, M_{N,2k+1}\,M_{N,2k-1}&=& (M_{N,2k})^2 -q^{-2k} (M_{N-1,k})^2\nonumber\\
q^N\, M_{N,2k+2}\,M_{N,2k}&=& (M_{N,2k+1})^2 -q^{N-1-(2k+1)} M_{N-1,k+1}\,M_{N-1,k}\label{BMsyslast}\\
\end{eqnarray}
\begin{eqnarray}
\hbox{Type $C_N$:}\qquad\qquad\qquad M_{a,k}\, M_{b,p}&=&q^{\Lambda_{a,b} p-\Lambda_{b,a}k}\, 
M_{b,p}\,M_{a,k}\nonumber \\
q^{2a} M_{a,k+1}\, M_{a,k-1}&=&M_{a,k}^2- M_{a+1,k}M_{a-1,k} \qquad (a\in[1,N-2])\nonumber \\
q^{2N-2} M_{N-1,2k+1}M_{N-1,2k-1}&=&M_{N-1,2k}^2 - q^{-Nk}\,M_{N-2,2k} M_{N,k}^2\nonumber  \\
q^{2N-2} M_{N-1,2k+2}M_{N-1,2k}&=&M_{N-1,2k+1}^2 - q^{-Nk}\,M_{N-2,2k+1} M_{N,k+1}M_{N,k}\nonumber  \\
q^N M_{N,k+1}\,M_{N,k-1}&=& M_{N,k}^2- q^{(N-2)k}\,M_{N-1,2k}\label{CMsys}
\end{eqnarray}
\end{conj}

In the case of $B_N$ we have not given an {\it a priori} definition of the order $N$ operators $M_{N,k}$.
They can be constructed as a suitable linear combination
of the limits at $t\to\infty$ of van Diejen operators \cite{vandiej,vandiejtwo}, compatible with the quantum Q-system in type B. 
These operators can be defined also by taking the quantum Q-system as the defining set of equations. First, define
$M_{N,2k}=\vert M(\{k,k,...,k\})\vert_{q^2}$ ($k$ repeated $N$ times), compatible with  Equation \eqref{BMsysthird},
in which $M_{N,2k}$ plays the role of the  quantum determinant of size $N$. Then Equation \eqref{BMsyslast}
gives information about $M_{N,2k+1}$:
$$(M_{N,2k+1})^2=q^N\, M_{N,2k+2}\,M_{N,2k}+q^{N-1-(2k+1)} M_{N-1,k+1}\,M_{N-1,k}.$$
This can be used to determine the relevant difference operators.
The fact that all the other equations of the system are satisfied is a highly non-trivial check.

These conjectures have been checked numerically up to $N=6$. We illustrate them in cases of small rank in Appendix A.

\begin{lemma}
The difference operators $M_{a,k}$ are a functional representation of the quantum Q-system  (\ref{DNqqsys}-\ref{CNqqsys}) up to the following rescaling:
\begin{eqnarray}D_N:\qquad\qquad\qquad
M_{a,k}&=& q^{\frac{a(a+1)}{2} -a(N+k)} \Q_{a,k}\qquad (\al\in[1,N-2])\nonumber\\
M_{N-1,k}&=& q^{-\frac{N(N-1)}{4}-\frac{Nk}{4}}\, \Q_{N-1,k}\nonumber\\
M_{N,k}&=& q^{-\frac{N(N-1)}{4}-\frac{Nk}{4}}\, \Q_{N,k}\qquad \qquad\qquad\qquad \qquad\qquad\qquad \qquad\qquad\label{Dch}\\
\nonumber\\
B_N:\qquad\qquad \qquad M_{a,k}&=& q^{\frac{a^2}{2}-a (N+k)} \,\Q_{a,k}\qquad (a\in[1,N-1])
\nonumber\\
M_{N,k}&=& q^{-\frac{N^2}{2}-\frac{Nk}{2}} \,\Q_{N,k}\label{Bch}\\
\nonumber\\
C_N:\qquad\qquad \qquad M_{a,k}&=&q^{\frac{a(a-1)}{2}-a(N+k)} \,\Q_{a,k}\qquad (a\in[1,N-1])
\nonumber\\
M_{N,k}&=&q^{-\frac{N(N+1)}{4}-\frac{Nk}{2}}\, \Q_{N,k}\label{Cch}
\end{eqnarray}
\end{lemma}
\begin{proof} By straightforward inspection.\end{proof}

\section{The raising/lowering operator conjectures for types $BCD$}

In this section, we present conjectures that extend the result of Theorem \ref{KNAN} to types $BCD$.
These involve the action of the difference operators $M_{a,k}$ on the dual $q$-Whittaker functions:
$$\Pi_\lambda^{(\g)}(q^{-1};\bx):=\lim_{t\to\infty} P_\lambda^{(\g)}(q,t;\bx).$$
The general idea is that while the operators $M_{a,0}^{(\g)}$ are all limits of Macdonald operators, for which the
$q$-Whittaker functions $\Pi_\lambda^{(\g)}$ are common eigenfunctions, the operators $M_{a,\pm 1}^{(\g)}$ are simple raising and lowering operators on those eigenfunctions.

\begin{conj}\label{knconj}
The operators $M_{a,0}^{(\g)}$ and $M_{a,\pm 1}^{(\g)}$ have the following action
on the $q$-Whittaker functions $\Pi^{(\g)}_\lambda$, valid for all $a\in [1,N]$:
\begin{eqnarray}
D_N, B_N:\qquad  M_{a,0}\, \Pi_\lambda&=&q^{2t_a (\lambda,\omega_a)}\, \Pi_{\lambda} 
\nonumber\\
M_{a,1}\, \Pi_\lambda&=& q^{2t_a(\lambda,\omega_a)}\, \Pi_{\lambda+\omega_a} \nonumber \\
M_{a,-1}\,  \Pi_\lambda&=&q^{2t_a(\lambda,\omega_a)} \big(1-q^{-2t_a(\al_a,\lambda)}\big)\, \Pi_{\lambda-\omega_a}
 \qquad
\label{DKN}\label{BKN}\\
\nonumber\\
C_N:\qquad  M_{a,0}\, \Pi_\lambda&=&q^{t_a(\omega_a,\lambda)}\,\Pi_\lambda \quad \nonumber\\
M_{a,1}\, \Pi_\lambda &=& q^{t_a\,(\omega_a,\lambda)} \, \Pi_{\lambda+\omega_a}  \quad \nonumber\\
M_{a,-1}\, \Pi_\lambda&=& q^{t_a(\omega_a,\lambda)}\big(1-q^{-t_a(\al_{a},\lambda)}\big) \, \Pi_{\lambda-\omega_a},\quad  \label{CKN}
\end{eqnarray}
where $\omega_a$ and $\al_a$ are the fundamental weights and simple roots of the corresponding algebras,  $t_a=2$ for the short roots in types $BC$, and $t_a=1$ otherwise.
\end{conj}
The eigenvalue equations (the first equation in each set) are a consequence of the Macdonald 
eigenvalue equations of the finite $t$ case. As in type $A$, up to a scalar multiple,
the raising equations (the second equation in each set) are a consequence of the quantum commutation relations
$M_{a,0}^{(\g)}\, M_{b,1}^{(\g)}=q^{\lambda_{a,b}^{(\g)}}  \, M_{b,1}^{(\g)}\, M_{a,0}^{(\g)}$. When applied
to the $q$-Whittaker functions, this gives: 
$$M_{a,0}^{(\g)}\Big( M_{b,1}^{(\g)}\, \Pi_\lambda^{(\g)}\Big)
=q^{\lambda_{a,b}^{(\g)}}\, M_{b,1}^{(\g)}\, M_{a,0}^{(\g)}\,  \Pi_\lambda^{(\g)}=
q^{\epsilon^{(\g)} \!t_a^{(\g)}\!(\omega_a^{(\g)}\!\!\!,\lambda+\omega_b^{(\g)}\!)}\, M_{b,1}^{(\g)}\,  \Pi_\lambda^{(\g)}$$
where we have defined $\epsilon^{(\g)}=2$ for $\g=B_N,D_N$ and $\epsilon^{(\g)}=1$
for $\g=C_N$, and used the fact that $\lambda_{a,b}^{(\g)}=\epsilon^{(\g)}t_a^{(\g)}(\omega_a^{(\g)},\omega_b^{(\g)})$. As the simulataneous eigenvalue property for the action of all $M_{a,0}$ 
determines the eigenvectors uniquely, up to a scalar multiple, we deduce that $M_{b,1}^{(\g)}\, \Pi_\lambda^{(\g)}$ must be proportional
to $\Pi_{\lambda+\omega_b}^{(\g)}$. The conjecture concerns simply the explicit value of the proportionality factor $q^{\epsilon^{(\g)}t_a^{(\g)}(\omega_a^{(\g)}\!\!,\lambda)}$, which we checked numerically up to $N=6$.

Similarly, the action of the lowering operator (third equation in each set) follows from the eigenvalue and raising operator actions and the quantum Q-system equations (\ref{DMsys}-\ref{CMsys})
of Conjecture \ref{qqconj}.
This is readily seen by applying the exchange relations involving $M_{a,1}$ and $M_{a,-1}$ 
to $\Pi_\lambda$.  

Examples of q-Whittaker functions and raising/lowering difference operators are given in Appendix A.

\section{The graded character conjectures for types $BCD$}

The main Conjecture \ref{qqconj} is a necessary ingredient in proving the following
conjecture about the expression of the graded characters in terms of difference operators. Let $I_<$ be the subset labels of the short roots the Dynkin diagram in types BC. Recall that $t_a=2$ for $a\in I_<$, and $t_a=1$ otherwise.

\begin{conj}\label{gradconj} 
The graded characters of the tensor products of KR-modules in types $BCD$ can be expressed using the iterated action of the difference operators in types $BCD$ on the polynomial 1:
$$
\chi_{\bn}^{(\g)}(q^{-1};\bx)=q^{-\frac{1}{2}Q^{(\g)}(\bn)}\, \prod_{\ell=k}^1\left(  \prod_{a\in I} (M_{ a,t_a\ell}^{(\g)})^{n_{a,t_a\ell}} \prod_{a\in I_<} (M_{a, t_a\ell-1}^{(\g)})^{n_a,t_a\ell-1}\right)\ \cdot 1
$$
where $k$ is chosen large enough to cover all the non-zero $n$'s, and
$$Q^{(\g)}(\bn):=\sum_{a,b=1}^N\sum_{i,j\geq 1} n_{a,i}\,\frac{\lambda_{a,b}^{(\g)} }{t_a^{(\g)}}\,
{\rm min}(t_b^{(\g)} i,t_a^{(\g)} j) \, n_{b,j} -\sum_{a=1}^N \sum_{i\geq 1} i\,  \lambda_{a,a}^{(\g)}\, n_{a,i}
$$
\end{conj}

The ordering of the difference operators in Conjecture \ref{gradconj} is consistent with \cite{qKR,Simon}, and is determined by the order of mutations in the bipartite belt of the quantum cluster algebra.

In particular, Conjecture \ref{gradconj} allows one to interpret the  so-called  ``level 1" graded characters 
corresponding to only possibly non-zero $n_a=n_{a,1}$, as the limiting Macdonald polynomials (or dual $q$-Whittaker functions)
$\Pi_\lambda$, with the correspondence:
\begin{equation}\label{dema} \chi_{\{n_a\}}^{(\g)}(q^{-1},\bx)=q^{-\frac{1}{2}Q^{(\g)}(\{n_a\})}\, \prod_{a\in I_>} \left(M_{a,1}^{(\g)}\right)^{n_a}\prod_{a\in I_<} \left(M_{a,1}^{(\g)}\right)^{n_a}  \cdot 1 =\Pi_\lambda^{(\g)}(q^{-1};\bx)
\end{equation}
where $I_>=I\setminus I_<$ is the set of long root labels, $\lambda=\sum_a n_a\omega_a$, and
the quadratic form $Q^{(\g)}$ is
\begin{eqnarray*}
Q^{(\g)}(\{n_a\}) &=&\sum_{a,b=1}^N n_{a}\,\frac{\lambda_{a,b}^{(\g)}}{t_a^{(\g)}} \min(t_a^{(\g)},t_b^{(\g)})\,n_{b} -\sum_{a=1}^N  \lambda_{a,a}^{(\g)}\, n_{a} \\
&=&\epsilon^{(\g)}\left\{ \sum_{a,b=1}^N n_{a}\,(\omega_a^{(\g)},\omega_b^{(\g)}) \min(t_a^{(\g)},t_b^{(\g)})\,n_{b} -\sum_{a=1}^N t_a^{(\g)} (\omega_a^{(\g)},\omega_a^{(\g)})\, n_{a}  \right\}  .
\end{eqnarray*}
Alternatively, for a fixed $\bn$, acting with $M_{a,k}^{(\g)}$ such that $t_a k \geq \max\{t_b j: n_{b,j}>0\}$ on $\chi_\bn$ gives
$$
M_{a,k}^{(\g)}\,\chi_\bn^{(\g)}(q^{-1};\bx) = q^{\epsilon^{(\g)} \sum_{j,b} (\omega_a^{(\g)},\omega_b^{(\g)}) \min(t^{(\g)}_b k,t^{(\g)}_a j)n_{b,j}} \chi_{\bn+\epsilon_{a,k}}^{(\g)}(q^{-1};\bx),
$$
where, again, $\bn=\sum n_{b,j} \epsilon_{b,j}$.

It is a non-trivial exercise to check the compatibility between \eqref{dema} and the raising operator
conditions on the second lines of (\ref{DKN}-\ref{CKN}). Indeed, we note that for a short root label $a$,
we may only act with $M_{a,1}^{(\g)}$ on characters obtained themselves by short root raising operators
say $\prod_{b\in I_<} (M_{b,1}^{(\g)})^{n_b}\cdot 1$:
\begin{eqnarray*}M_{a,1}^{(\g)}\, \chi_{\{n_b\}}^{(\g)}(q^{-1},\bx)&=&q^{-\half Q^{(\g)}(\{n_b\})} \prod_{b\in I_<} 
(M_{b,1}^{(\g)})^{n_b+\delta_{a,b}}\cdot 1\\
&=&q^{\half(Q^{(\g)}(\{n_b+\delta_{a,b}\})-Q^{(\g)}(\{n_b\}))} \chi_{\{n_b+\delta_{a,b}\}}^{(\g)}(q^{-1},\bx)
\end{eqnarray*}
We compute:
\begin{eqnarray*}Q^{(\g)}(\{n_b+\delta_{a,b}\})-Q^{(\g)}(\{n_b\})&=&\epsilon^{(\g)} \sum_{b\in I_<} n_b \min(t_a^{(\g)},t_b^{(\g)}) (\omega_a^{(\g)},\omega_b^{(\g)})\\
&=&
\epsilon^{(\g)}\, t_a^{(\g)} (\omega_a^{(\g)},\sum_b n_b \omega_b^{(\g)})=\epsilon^{(\g)}\, t_a^{(\g)} (\omega_a^{(\g)},\lambda)
\end{eqnarray*}
by use of $t_a^{(\g)}=t_b^{(\g)}=2$. This is
in agreement with (\ref{DKN}-\ref{CKN}), upon writing $\lambda=\sum n_b \omega_b$.
For a long root label $a$, we have similarly:
\begin{eqnarray*}Q^{(\g)}(\{n_b+\delta_{a,b}\})-Q^{(\g)}(\{n_b\})&=&\epsilon^{(\g)} \sum_{b\in I} n_b \min(t_a^{(\g)},t_b^{(\g)}) (\omega_a^{(\g)},\omega_b^{(\g)})\\
&=&
\epsilon^{(\g)}\,  (\omega_a^{(\g)},\sum_b n_b \omega_b^{(\g)})=\epsilon^{(\g)}\, t_a^{(\g)} (\omega_a,\lambda)
\end{eqnarray*}
by use of $t_a^{(\g)}=1$. This is again in agreement with (\ref{DKN}-\ref{CKN}).

This special case is also in agreement with \cite{LNSSS}.

\section{Conclusion}

\subsection{Summary: Macdonald operators and quantum cluster algebra}
In this paper we have presented two main conjectures about the difference operators $M_{a,k}^{(\g)}$ corresponding to  root systems of  types $BCD$. 

The first conjecture states that the difference operators $\{M_{a,k}^{(\g)}\}$ satisfy  renormalized  quantum Q-systems of types $BCD$. The commuting difference operators $\{M_{a,0}^{(\g)}\}$ can be obtained as the limits $t\to \infty$  of appropriate Macdonald operators for the relevant Lie root systems.  The operators $M_{a,k}^{(\g)}$ are their $SL_2(\Z)$ ``discrete time evolution".

The second conjecture is that the  difference operators $M_{a,\pm 1}^{(\g)}$
at times $k=\pm 1$ act as 
raising and lowering operators on $q$-Whittaker functions  $\Pi_\lambda^{(\g)}$.

The quantum Q-systems are mutations in the corresponding quantum cluster algebras.
From this point of view, the sets $S_{\pm}:=\{ M_{a,0}^{(\g)},M_{a,\pm 1}^{(\g)}\}$
are two possible valid initial cluster seeds, and the conjectures state that, as in type $A$, these are formed
of the Macdonald and raising (reps. lowering) operators, at $t\to \infty$. This raises a number of questions
regarding cluster variables in general: from preliminary inspection, 
it appears that {\it all cluster
variables} in the corresponding cluster algebra are difference operators as well, a quite surprising property
which goes way beyond the usual Laurent property of quantum cluster algebra, which would only imply that
all cluster variables are Laurent polynomials of those in an initial cluster. These other difference operators should tell us something new about Macdonald theory (including in type $A$).

\subsection{Towards proving the conjectures}

The second conjecture of this paper offers a possible strategy for proving both conjectures which goes as follows. The second conjecture indeed can be restated as follows: the Macdonald and raising operators $M_{a,0}^{(\g)}$ and $M_{a,1}^{(\g)}$
act on the basis $\Pi_\lambda$ of Weyl-symmetric polynomials of $\bx$ as very simple {\it dual} 
operators $X_a,P_a$:
$$M_{a,0}^{(\g)}\, \Pi_\lambda^{(\g)}=  \Pi_\lambda^{(\g)} \, X_a,
\quad M_{a,1}^{(\g)}\, \Pi_\lambda^{(\g)}=  \Pi_\lambda^{(\g)} \, P_a $$
where the operators $X_a,P_a$ act to the left on the variables $\Lambda_a=q^{2\lambda_a}$.
More precisely the operators $X_a$ acts diagonally on the basis $\Pi_\lambda^{(\g)}$ with eigenvalues $q^{\epsilon_a^{(\g)}(\omega_a^{(\g)}\!\!,\lambda)}$ whereas $P_a$ includes a 
multiplicative shift of $\Lambda$ variables. As a result, the operators $X_a,P_a$ obey the simple (opposite)  commutation relations:  $X_a \, P_b=q^{-\lambda_{a,b}^{(\g)}}\,P_b\, X_a$.

Reversing the logic entirely, we may {\it start from the data} of operators $X_a,P_a$, and {\it define} the left action
of the operators $M_{a,k}^{(\g)}$ via the (opposite)
renormalized quantum Q-systems of this paper. To identify the left action and the right one, we need
to construct the Gaussian operator in $X_a,P_a$ variables, and the conjectures will follow from the anti-homomorphism mapping left and right actions. We will pursue this program elsewhere.

\subsection{sDAHA, EHA and the $t$-deformation of quantum Q-systems}

In \cite{DFKqt}, we investigated the natural $t$-deformation of the $A$ type quantum Q-system provided
by the type $A$ sDAHA, also expressed as a quotient of the elliptic Hall algebra (EHA) \cite{SCHIFFEHA}, or the
quantum toroidal algebra of ${\gl}_1$. There exist sDAHAs for all classical types \cite{Cheredbook}, and this would be the natural candidate for a generalization of the $A$ type results.
We may use as starting points the $(q,t)$-Macdonald operators of Section \ref{diff_ops}, and their time evolution via the suitable $SL_2(\Z)$ action, to derive current algebra relations that will generalize
EHA or quantum toroidal algebra relations, ideally giving rise to new interesting algebras.

\begin{appendix}
\section{Examples}

This appendix gathers a few inter-connected examples, which give non-trivial checks of the conjectures of this paper.
We give here the examples of the $B_2$ and $C_2$ cases, that of the $A_3$ and $D_3$,
and finally the $D_4$ case. We verify in particular that the symmetries of the Dynkin diagrams are also reflected on our difference operators.

\subsection{Weyl-invariant Schur functions}

\begin{defn}
We define the Weyl-invariant Schur functions 
$s^{({\mathfrak g})}_\lambda(\bx)$ for ${\mathfrak g}=A_{N-1},B_N,C_N,D_N$ to be the characters of the irreducible representations
of corresponding weight \cite{Mizu}:\footnote{In \cite{Mizu}, the author restricts the definition to $\lambda$'s that are actual partitions. Here we include all possible dominant weights.}:
\begin{equation}
 s^{(A_{N-1})}_\lambda(\bx)=
\frac{\det_{1\leq i,j\leq N}\left( x_i^{N-j+\lambda_j}\right)}
{\det_{1\leq i,j\leq N}\left(  x_i^{N-j}\right)} \label{ANschur}
\end{equation}
\begin{equation}
s^{(B_N)}_\lambda(\bx)=
\frac{\det_{1\leq i,j\leq N}\left( x_i^{N-j+\lambda_j+\frac{1}{2}}-\frac{1}{x_i^{N-j+\lambda_j+\frac{1}{2}}}\right)}
{\det_{1\leq i,j\leq N}\left(  x_i^{N-j+\frac{1}{2}}-\frac{1}{x_i^{N-j+\frac{1}{2}}}\right)} \label{BNschur} 
\end{equation}
\begin{equation}
 s^{(C_N)}_\lambda(\bx)=\frac{\det_{1\leq i,j\leq N}\left( x_i^{N-j+\lambda_j+1}-\frac{1}{x_i^{N-j+\lambda_j+1}}\right)}
{\det_{1\leq i,j\leq N}\left(  x_i^{N-j+1}-\frac{1}{x_i^{N-j+1}}\right)} \label{CNschur}
\end{equation}
\begin{equation}
 s^{(D_N)}_\lambda(\bx)= 
\frac{\det_{1\leq i,j\leq N}\left( x_i^{N-j+\lambda_j}-\frac{1}{x_i^{N-j+\lambda_j}}\right)}
{\det_{1\leq i,j\leq N}\left(  x_i^{N-j}+\frac{1}{x_i^{N-j}}\right)}
+\frac{\det_{1\leq i,j\leq N}\left( x_i^{N-j+\lambda_j}+\frac{1}{x_i^{N-j+\lambda_j}}\right)}
{\det_{1\leq i,j\leq N}\left(  x_i^{N-j}+\frac{1}{x_i^{N-j}}\right)}\label{DNschur} 
\end{equation}
\end{defn}

The Weyl-invariant Schur functions form a basis of the space of Weyl-invariant (Laurent) polynomials.
In the following sections, we write the dual $q$-Whittaker functions $\Pi_\lambda$ in this Schur basis. We drop the superscript $({\mathfrak g})$ for simplicity.

\begin{remark}
Note that  Equation \eqref{ANschur} for ${\mathfrak gl}_N$ Schur functions
gives the ${\mathfrak sl}_N$ Schur functions, upon restriction to $x_1x_2\cdots x_N=1$ and 
then noting that $s_\lambda=s_{\lambda-\rho}$ where $\rho=e_1+e_2+\cdots +e_N$.
\end{remark}

\subsection{The $B_2$ case}

\subsubsection{M operators}

\begin{eqnarray*}M_{1,k}&=&\sum_{\epsilon=\pm 1}
x_1^{k \epsilon}\, \frac{x_1^\epsilon}{x_1^\epsilon- 1} 
\frac{x_1^\epsilon x_2}{x_1^\epsilon x_2 - 1}\,  \frac{x_1^\epsilon}{x_1^\epsilon - x_2} 
\, \Gamma_1^{2\epsilon}
+x_2^{k \epsilon}\, \frac{x_2^\epsilon}{x_2^\epsilon- 1} 
\frac{x_2^\epsilon x_1}{x_2^\epsilon x_1 - 1}\,  \frac{x_2^\epsilon}{
 x_2^\epsilon - x_1} \, \Gamma_2^{2\epsilon}\\
M_{2,2k}&=&q^{-2k}+\sum_{\epsilon_1,\epsilon_2=\pm 1}
\frac{x_1^{\epsilon_1} }{x_1^{\epsilon_1}-1}\, \frac{x_2^{\epsilon_2} }{x_2^{\epsilon_2}-1}\, \frac{x_1^{\epsilon_1}x_2^{\epsilon_2} }{x_1^{\epsilon_1}x_2^{\epsilon_2}-1}\,
\frac{q^2x_1^{\epsilon_1}x_2^{\epsilon_2} }{q^2x_1^{\epsilon_1}x_2^{\epsilon_2}-1}
(x_1^{k\epsilon_1}x_2^{k\epsilon_2}\,\Gamma_1^{2\epsilon_1}\Gamma_2^{2\epsilon_2}-q^{-2k})\\
M_{2,2k-1}&=&\sum_{\epsilon_1,\epsilon_2=\pm 1}
\frac{x_1^{\epsilon_1} }{x_1^{\epsilon_1}-1}\, \frac{x_2^{\epsilon_2} }{x_2^{\epsilon_2}-1}\, \frac{x_1^{\epsilon_1}x_2^{\epsilon_2} }{x_1^{\epsilon_1}x_2^{\epsilon_2}-1}\,
\frac{q^2x_1^{\epsilon_1}x_2^{\epsilon_2} }{q^2x_1^{\epsilon_1}x_2^{\epsilon_2}-1}
(x_1^{\epsilon_1}x_2^{\epsilon_2})^{-\frac{1}{2}}\, (x_1^{k\epsilon_1}x_2^{k\epsilon_2}\,\Gamma_1^{2\epsilon_1}\Gamma_2^{2\epsilon_2}-q^{-2k})
\end{eqnarray*}
Here $M_{1,k}$ is given by the $k$-th iterate conjugation of \eqref{macBN} w.r.t. the Gaussian, and
$M_{2,2k}$ by the quantum determinant $M_{1,k}^2-q^2M_{1,k+1}M_{1,k-1}$.

\subsubsection{Dual $q$-Whittaker functions}

\begin{eqnarray*}
\Pi_{0,0}&=& s_{0,0}\\
\Pi_{1,0}&=& s_{1,0} \\
\Pi_{2,0}&=& s_{2,0}+q^{-2}s_{1,1}+q^{-4}s_{0,0} \\
\Pi_{1,1}&=& s_{1,1}+q^{-2}s_{1,0}+q^{-2}s_{0,0} \\
\Pi_{3,0}&=& s_{3,0}+\frac{1+q^2}{q^4} s_{2,1}+q^{-6}s_{1,1}+\frac{1+q^2+q^4}{q^8}s_{1,0}\\
\Pi_{2,1}&=& s_{2,1}+q^{-2}s_{2,0}+\frac{1+q^2}{q^4} s_{1,1}+\frac{1+q^2}{q^4} s_{1,0}+q^{-6}s_{0,0} \\
\Pi_{\frac{1}{2},\frac{1}{2}}&=& s_{\frac{1}{2},\frac{1}{2}}\\
\end{eqnarray*}
\begin{eqnarray*}
\Pi_{\frac{3}{2},\frac{1}{2}}&=& s_{\frac{3}{2},\frac{1}{2}}+q^{-2}s_{\frac{1}{2},\frac{1}{2}}\\
\Pi_{\frac{5}{2},\frac{1}{2}}&=& s_{\frac{5}{2},\frac{1}{2}}+q^{-2}s_{\frac{3}{2},\frac{3}{2}}+\frac{1+q^2}{q^4} s_{\frac{3}{2},\frac{1}{2}}+
\frac{1+q^2}{q^6}s_{\frac{1}{2},\frac{1}{2}}\\
\Pi_{\frac{3}{2},\frac{3}{2}}&=& s_{\frac{3}{2},\frac{3}{2}}+\frac{1+q^2}{q^4} s_{\frac{3}{2},\frac{1}{2}}+
\frac{1+q^2+q^4}{q^6}s_{\frac{1}{2},\frac{1}{2}}
\end{eqnarray*}

\subsection{The $C_2$ case}

\subsubsection{The $B_2\leftrightarrow C_2$ symmetry}

The $B_2$ case can be mapped onto the $C_2$ case, by interchanging the roles of the two fundamental weights. More precisely, let us denote by $\omega_a=\omega_a^{(B_2)}$ (resp. $\omega_a'=\omega_a^{(C_2)}$). The variables $x_1,x_2$ of the $B_2$ case can be thought of as
$x_i=e^{e_i}$, with  $e_1=\omega_1$ and $e_2=2\omega_2-\omega_1$,
and similarly for $C_2$, where $x_i'=e^{e_i'}$, $e_1'=\omega_1'$ and $e_2'=\omega_2'-\omega_1'$. The mapping 
$(\omega_1,\omega_2)\mapsto (\omega_2',\omega_1')$ sends
\begin{equation}\label{mapx} x_1=e^{\omega_1}\mapsto e^{\omega_2'}=x_1'\, x_2', \quad x_2=e^{2\omega_2-\omega_1}\mapsto e^{2\omega_1'-\omega_2'}=\frac{x_1'}{x_2'}
\end{equation}
Similarly, we have:
\begin{equation}\label{mapga} \Gamma_1^2 \mapsto \Gamma_1'\, \Gamma_2',\qquad \Gamma_2^2\mapsto \Gamma_1'\, {\Gamma_2'}^{-1}
\end{equation}
The map (\ref{mapx}-\ref{mapga}) sends the operators $M_{a,k}^{(B_2)}\mapsto M_{3-a,k}^{(C_2)}$ for $a=1,2$ and $k\in \Z$,
and the Macdonald polynomials $P_{\lambda_1,\lambda_2}^{(B_2)}\mapsto P_{\lambda_1+\lambda_2,\lambda_1-\lambda_2}^{(C_2)}$ and similarly for the $q$-Whittaker functions,
using
$$ \lambda=\lambda_1 e_1+\lambda_2 e_2=(\lambda_1-\lambda_2)\omega_1+2\lambda_2 \omega_2\mapsto 
(\lambda_1-\lambda_2)\omega_2'+2\lambda_2 \omega_1'=(\lambda_1+\lambda_2)e_1'+(\lambda_1-\lambda_2)e_2' .$$

\subsubsection{M operators}
From the definitions in type $C$,
\begin{eqnarray*}
M_{1,2k}&=&q^{-2k}+\sum_{\epsilon=\pm 1} \left\{
\frac{x_1^{2\epsilon}}{x_1^{2\epsilon}-1}\frac{q^2x_1^{2\epsilon}}{q^2x_1^{2\epsilon}-1} 
\frac{x_1^\epsilon x_2}{x_1^\epsilon x_2 - 1}\frac{x_1^\epsilon }{x_1^\epsilon -x_2 }
(x_1^{2k \epsilon} \Gamma_1^{2\epsilon} -q^{-2k}) \right.\\
&&\qquad\qquad +\left. \frac{x_2^{2\epsilon}}{x_2^{2\epsilon}-1}\frac{q^2x_2^{2\epsilon}}{q^2x_2^{2\epsilon}-1} 
\frac{x_2^\epsilon x_1}{x_2^\epsilon x_1 - 1}\frac{x_2^\epsilon }{x_2^\epsilon -x_1 }
(x_2^{2k \epsilon} \Gamma_2^{2\epsilon} -q^{-2k})\right\},\\
M_{1,2k-1}&=&\sum_{\epsilon=\pm 1} \left\{
\frac{x_1^{2\epsilon}}{x_1^{2\epsilon}-1}\frac{q^2x_1^{2\epsilon}}{q^2x_1^{2\epsilon}-1} 
\frac{x_1^\epsilon x_2}{x_1^\epsilon x_2 - 1}\frac{x_1^\epsilon }{x_1^\epsilon -x_2 }\, x_1^{-\epsilon}
(x_1^{2k \epsilon} \Gamma_1^{2\epsilon} -q^{-2k})\right.\\
&&\qquad+\left. \frac{x_2^{2\epsilon}}{x_2^{2\epsilon}-1}\frac{q^2x_2^{2\epsilon}}{q^2x_2^{2\epsilon}-1} 
\frac{x_2^\epsilon x_1}{x_2^\epsilon x_1 - 1}\frac{x_2^\epsilon }{x_2^\epsilon -x_1 }\, x_2^{-\epsilon}
(x_2^{2k \epsilon} \Gamma_2^{2\epsilon} -q^{-2k})\right\},\\
M_{2,k}&=& \sum_{\epsilon_1,\epsilon_2=\pm 1} \frac{x_1^{2\epsilon_1}}{x_1^{2\epsilon_1}-1}\frac{x_2^{2\epsilon_2}}{x_2^{2\epsilon_2}-1}
\frac{x_1^{\epsilon_1}x_2^{\epsilon_2} }{x_1^{\epsilon_1}x_2^{\epsilon_2}-1}\, 
(x_1^{k\epsilon_1} x_2^{k\epsilon_2})\, \Gamma_1^{\epsilon_1}\Gamma_2^{\epsilon_2}.
\end{eqnarray*}

The expected symmetry between $M_{1,k}^{(B_2)}\mapsto M_{2,k}^{(C_2)}$ and 
$M_{2,k}^{(B_2)}\mapsto M_{1,k}^{(C_2)}$ is easily checked.

\subsubsection{Dual $q$-Whittaker functions}

\begin{eqnarray*}
\Pi_{0,0}&=& s_{0,0}\\
\Pi_{1,0}&=& s_{1,0} \\
\Pi_{2,0}&=& s_{2,0}+q^{-2}s_{1,1}+q^{-2}s_{0,0} \\
\Pi_{1,1}&=& s_{1,1} \\
\Pi_{3,0}&=& s_{3,0}+\frac{1+q^2}{q^4} s_{2,1}+\frac{1+q^2+q^4}{q^6}s_{1,0}\\
\Pi_{2,1}&=& s_{2,1}+q^{-2}s_{2,0}+q^{-2} s_{1,0}\\
\Pi_{4,0}&=& s_{4,0}+\frac{1+q^2+q^4}{q^6} s_{3,1}+\frac{1+q^4}{q^8}s_{2,2}+\frac{(1+q^2+q^4)(1+q^4)}{q^{10}}s_{2,0}\\
&&+\frac{1+2q^2+q^4+q^6}{q^{10}}s_{1,1}+\frac{1+q^4+q^8}{q^{12}}s_{0,0}\\
\Pi_{3,1}&=& s_{3,1}+q^{-2}s_{2,2}+\frac{(1+q^2)}{q^{4}}s_{2,0}+\frac{1+q^2}{q^{4}}s_{1,1}+q^{-6}s_{0,0}\\
\Pi_{2,2}&=&s_{2,2}+q^{-2}s_{2,0}+q^{-4}s_{0,0}
\end{eqnarray*}

The expected symmetry relations between $B_2$ and $C_2$ $q$-Whittaker functions are easily checked, using the explicit expressions for the relevant Schur functions (\ref{BNschur}-\ref{CNschur}).

\subsection{The ${\mathfrak sl}_4$ and $D_3$ cases}

\subsubsection{The ${\mathfrak sl}_4\leftrightarrow D_3$ symmetry}
Compared to the ${\mathfrak gl}_4$ case, the symmetric functions of the case ${\mathfrak sl}_4$ involve the extra condition that $x_1x_2x_3x_4=1$, and accordingly $\Gamma_1\Gamma_2\Gamma_3\Gamma_4=1$
(see Remark \ref{glvssl}).
This is implemented by imposing the extra condition $e_1+e_2+e_3+e_4=0$ under which:
$$ e_1=\omega_1,\ e_2=\omega_2-\omega_1,\ e_3=\omega_3-\omega_2,\ e_4=-\omega_3 .$$
Primed variables are used for $D_3$:
$$ e_1'=\omega_1',\ e_2'=\omega_2'+\omega_3'-\omega_1',\ e_3'=\omega_3'-\omega_2' .$$
We use the mapping
$$\omega_1\mapsto \omega_3', \quad \omega_2\mapsto \omega_1', \quad\omega_3\mapsto \omega_1' .$$
This is equivalent to the changes of variables (using $x_i=e^{e_i}, x_i'=e^{e_i'}$):
\begin{equation}\label{sl4tod3} x_1\mapsto \sqrt{x_1'x_2'x_3'},\quad x_2\mapsto \sqrt{\frac{x_1'}{x_2'x_3'}},\quad x_3\mapsto\sqrt{\frac{x_2'}{x_1'x_3'}},\quad x_4\mapsto \sqrt{\frac{x_3'}{x_1'x_2'}} .
\end{equation}
Moreover, to account for our choice of normalization we must also take $q\mapsto q^2$, which results
in
\begin{equation}\label{sl4toD3} \Gamma_1 \mapsto \Gamma_1'\Gamma_2'\Gamma_3',\quad \Gamma_2\mapsto  \frac{\Gamma_1'}{\Gamma_2'\Gamma_3'}, \quad \Gamma_3\mapsto  \frac{\Gamma_2'}{\Gamma_1'\Gamma_3'}, \quad \Gamma_4\mapsto  \frac{\Gamma_3'}{\Gamma_1'\Gamma_2'} .
\end{equation}

The above transformations send the ${\mathfrak sl}_4$ Macdonald operators to the $D_3$ ones,
namely $M_{1,k}\mapsto M_{3,k}'$, $M_{2,k}\mapsto M_{1,k}'$ and $M_{3,k}\mapsto M_{2,k}'$.
The corresponding mapping of Macdonald polynomials is
\begin{equation}
\label{maPsl4tod3}
P_{\lambda_1,\lambda_2,\lambda_3,\lambda_4}^{({\mathfrak sl}_4)}\mapsto P_{\frac{\lambda_1+\lambda_2-\lambda_3-\lambda_4}{2},\frac{\lambda_1-\lambda_2+\lambda_3-\lambda_4}{2},\frac{\lambda_1-\lambda_2-\lambda_3+\lambda_4}{2}}^{(D_3)} .
\end{equation}

\subsubsection{M operators}

The ${\mathfrak sl}_4$ operators are:
\begin{eqnarray*}
M_{1,k}^{({\mathfrak sl}_4)}&=&\sum_{i=1}^4 x_i^k\, \prod_{j\neq i} \frac{x_i}{x_i-x_j} \Gamma_i \\
M_{2,k}^{({\mathfrak sl}_4)}&=&\sum_{1\leq i<j\leq 4}^4 (x_i x_j)^k\, \prod_{m\neq i,j}  \frac{x_m}{x_m-x_i}\frac{x_m}{x_m-x_j}\, \Gamma_i\Gamma_j\\
M_{3,k}^{({\mathfrak sl}_4)}&=&\sum_{i=1}^4 x_i^{-k}\, \prod_{j\neq i} \frac{x_j}{x_j-x_i} \Gamma_i^{-1}
\end{eqnarray*}
where $x_1x_2x_3x_4=1$ and $\Gamma_1\Gamma_2\Gamma_3\Gamma_4=1$.

The $D_3$ operators are:
\begin{eqnarray*}
M_{1,k}^{(D_3)}&=&\sum_{\epsilon=\pm 1} \sum_{i=1}^3 x_i^k\, \prod_{j\neq i}\frac{x_i^\epsilon x_j}{x_i^\epsilon x_j-1} \frac{x_i^\epsilon}{x_i^\epsilon-x_j}\Gamma_i^{2\epsilon} \\
M_{2,k}^{(D_3)}&=&\sum_{\epsilon_1,\epsilon_2,\epsilon_3=\pm 1\atop
\epsilon_1\epsilon_2\epsilon_3=-1} \prod_{i=1}^3 x_i^{k\epsilon_i}\, \prod_{1\leq i<j\leq 3} \frac{x_i^{\epsilon_i} x_j^{\epsilon_j}}{x_i^{\epsilon_i} x_j^{\epsilon_j}-1} \prod_{i=1}^3 \Gamma_i^{\epsilon_i}\\
M_{3,k}^{(D_3)}&=&\sum_{\epsilon_1,\epsilon_2,\epsilon_3=\pm 1\atop
\epsilon_1\epsilon_2\epsilon_3=1} \prod_{i=1}^3 x_i^{k\epsilon_i}\, \prod_{1\leq i<j\leq 3} \frac{x_i^{\epsilon_i} x_j^{\epsilon_j}}{x_i^{\epsilon_i} x_j^{\epsilon_j}-1} \prod_{i=1}^3 \Gamma_i^{\epsilon_i}
\end{eqnarray*}

It is straightforward to see that the change of variables (\ref{sl4tod3}-\ref{sl4toD3}) map the M operators as follows:  $M_{1,k}^{({\mathfrak sl}_4)}\mapsto M_{3,k}^{(D_3)}$, $M_{2,k}^{({\mathfrak sl}_4)}\mapsto M_{1,k}^{(D_3)}$ and $M_{3,k}^{({\mathfrak sl}_4)}\mapsto M_{2,k}^{(D_3)}$.

\subsubsection{Dual $q$-Whittaker functions}

In terms of the Schur functions \eqref{ANschur}, the first few $A_3$ dual $q$-Whittaker functions are:
\begin{eqnarray*}
\Pi_{0,0,0,0}&=&1\\
\Pi_{1,0,0,0}&=&s_{1,0,0,0}\\
\Pi_{2,0,0,0}&=&s_{2,0,0,0}+q^{-1}\, s_{1,1,0,0}\\
\Pi_{1,1,0,0}&=&s_{1,1,0,0}\\
\Pi_{3,0,0,0}&=&s_{3,0,0,0}+\frac{1+q}{q^2}\, s_{2,1,0,0}+ q^{-3}\, s_{1,1,1,0}\\
\Pi_{2,1,0,0}&=&s_{2,1,0,0}+q^{-1}\, s_{1,1,1,0}\\
\Pi_{1,1,1,0}&=&s_{1,1,1,0}\\
\Pi_{4,0,0,0}&=&s_{4,0,0,0}+\frac{1+q+q^2}{q^3}s_{3,1,0,0}+\frac{1+q^2}{q^4}s_{2,2,0,0}+\frac{1+q+q^2}{q^5}s_{2,1,1,0}+q^{-6}s_{1,1,1,1}\\
\Pi_{3,1,0,0}&=&s_{3,1,0,0}+q^{-1}s_{2,2,0,0}+\frac{1+q}{q^2}s_{2,1,1,0}+q^{-3}s_{1,1,1,1}\\
\Pi_{2,2,0,0}&=&s_{2,2,0,0}+q^{-1}s_{2,1,1,0}+q^{-2}s_{1,1,1,1}\\
\Pi_{2,1,1,0}&=&s_{2,1,1,0}+q^{-1}s_{1,1,1,1}\\
\Pi_{1,1,1,1}&=&s_{1,1,1,1}
\end{eqnarray*}

In terms of the Schur functions \eqref{DNschur}, the first few $D_3$ dual $q$-Whittaker functions are:
\begin{eqnarray*}
\Pi_{0,0,0}&=&1\\
\Pi_{1,0,0}&=&s_{1,0,0}\\
\Pi_{2,0,0}&=&s_{2,0,0}+q^{-2}\, s_{1,1,0}+q^{-4}\,s_{0,0,0}\\
\Pi_{1,1,0}&=&s_{1,1,0}+q^{-2}\,s_{0,0,0}\\
\Pi_{3,0,0}&=&s_{3,0,0}+\frac{1+q^2}{q^4}\, s_{2,1,0}+ q^{-6}\, (s_{1,1,1}+s_{1,1,-1})+\frac{1+q^2+q^4}{q^8}\,s_{1,0,0}\\
\Pi_{2,1,0}&=&s_{2,1,0}+q^{-2}\, (s_{1,1,1}+s_{1,1,-1})+\frac{1+q^2}{q^4}\,s_{1,0,0}\\
\Pi_{1,1,\epsilon}&=&s_{1,1,\epsilon}+q^{-2}\, s_{1,0,0}\\
\Pi_{\frac{1}{2},\frac{1}{2},\frac{\epsilon}{2}}&=&s_{\frac{1}{2},\frac{1}{2},\frac{\epsilon}{2}}\\
\Pi_{\frac{3}{2},\frac{1}{2},\frac{\epsilon}{2}}&=&s_{\frac{3}{2},\frac{1}{2},\frac{\epsilon}{2}}+q^{-2}\,
s_{\frac{1}{2},\frac{1}{2},-\frac{\epsilon}{2}}\\
\Pi_{\frac{5}{2},\frac{1}{2},\frac{\epsilon}{2}}&=&s_{\frac{5}{2},\frac{1}{2},\frac{\epsilon}{2}}+q^{-2}\,s_{\frac{3}{2},\frac{3}{2},\frac{\epsilon}{2}}+\frac{1+q^2}{q^4}\,
s_{\frac{3}{2},\frac{1}{2},-\frac{\epsilon}{2}}+\frac{1+q^2}{q^6}\,s_{\frac{1}{2},\frac{1}{2},\frac{\epsilon}{2}}\\
\end{eqnarray*}
\begin{eqnarray*}
\Pi_{\frac{3}{2},\frac{3}{2},\frac{\epsilon}{2}}&=&s_{\frac{3}{2},\frac{3}{2},\frac{\epsilon}{2}}+q^{-2}\,
s_{\frac{3}{2},\frac{1}{2},-\frac{\epsilon}{2}}+\frac{1+q^2}{q^4}\,s_{\frac{1}{2},\frac{1}{2},\frac{\epsilon}{2}}\\
\Pi_{\frac{3}{2},\frac{3}{2},\frac{3\epsilon}{2}}&=&s_{\frac{3}{2},\frac{3}{2},\frac{3\epsilon}{2}}+\frac{1+q^2}{q^4}\, s_{\frac{3}{2},\frac{1}{2},\frac{\epsilon}{2}}+q^{-6}\, s_{\frac{1}{2},\frac{1}{2},-\frac{\epsilon}{2}}
\end{eqnarray*}

The relations \eqref{maPsl4tod3} are easily checked for these polynomials. First we note that
the mapping \eqref{maPsl4tod3} extends to the Weyl-invariant Schur functions as well. We then simply have to
check that the coefficients in the $D_3$ case match those in the ${\mathfrak sl}_4$ case up to $q\to q^2$.
For instance, we have $\Pi_{2,1,0,0}^{({\mathfrak sl}_4)}\mapsto \Pi_{\frac{3}{2},\frac{1}{2},\frac{1}{2}}^{(D_3)}$
as the coefficient of $s_{1,1,1,0}^{({\mathfrak sl}_4)}$ (resp. $s_{\frac{1}{2},\frac{1}{2},-\frac{\epsilon}{2}}^{(D_3)}$) is $q^{-1}$ (resp. $q^{-2}$). Similarly the coefficients in $\Pi_{3,0,0,0}^{({\mathfrak sl}_4)}$ and $\Pi_{\frac{3}{2},\frac{3}{2},\frac{3}{2}}^{(D_3)}$ also agree up to $q\to q^2$.

\subsection{The $D_4$ case}

\subsubsection{Symmetries}

There are two simple symmetries of the Dynkin diagram which induce symmetries of the Macdonald operators and polynomials:
\begin{itemize}
\item 
The $\Z_2$ automorphism of the Dynkin diagram under which $1\to 1$, $2\to 2$ and $3\leftrightarrow 4$,
\item
The $\Z_3$ automorphism of the Dynkin diagram under which $1\to 3\to 4\to 1$ and $2\to 2$. 
\end{itemize}
Recall the relations
$$e_1=\omega_1,\ e_2=\omega_2-\omega_1,\ e_3=\omega_4+\omega_3-\omega_2,\ e_4=\omega_4-\omega_3 .$$
The $\Z_2$ symmetry under:
$$ \omega_1\mapsto \omega_1,\ \omega_2\mapsto \omega_2,\ \omega_3\mapsto \omega_4,\ 
\omega_4\mapsto \omega_3 $$
induces the transformations
\begin{equation}
\label{mapxDtwo}
x_1\mapsto x_1,\ x_2\mapsto x_2,\ x_3\mapsto
x_3,\ x_4\mapsto \frac{1}{x_4} .
\end{equation}
Similarly, we have
$$ \Gamma_1\mapsto \Gamma_1,\ 
 \Gamma_2\mapsto \Gamma_2,\ 
  \Gamma_3\mapsto \Gamma_3,\ 
   \Gamma_4 \mapsto \Gamma_4^{-1} .$$
We will check that under these transformations, we have 
$M_{1,k}\mapsto M_{1,k}$, $M_{2,k}\mapsto M_{2,k}$, $M_{3,k}\mapsto M_{4,k}$ and $M_{4,k}\mapsto M_{3,k}$
for Macdonald operators, and 
$$ P_{\lambda_1,\lambda_2,\lambda_3,\lambda_4}\mapsto P_{\lambda_1,\lambda_2,\lambda_3,-\lambda_4} 
$$
for Macdonald polynomials as well as $q$-Whittaker functions.

The $\Z_3$ symmetry sends
$$ \omega_1\mapsto \omega_3,\ \omega_2\mapsto \omega_2,\ \omega_3\mapsto \omega_4,\ 
\omega_4\mapsto \omega_1 $$
therefore induces the transformations
\begin{equation}
\label{mapxD}
x_1\mapsto \sqrt{\frac{x_1x_2x_3}{x_4}},\ x_2\mapsto\sqrt{\frac{x_1x_2x_4}{x_3}},\ x_3\mapsto
\sqrt{\frac{x_1x_3x_4}{x_2}},\ x_4\mapsto \sqrt{\frac{x_1}{x_2x_3x_4}} .
\end{equation}
Similarly, we have
$$ \Gamma_1^2\mapsto \frac{\Gamma_1\Gamma_2\Gamma_3}{\Gamma_4},\ 
 \Gamma_2^2\mapsto \frac{\Gamma_1\Gamma_2\Gamma_4}{\Gamma_3},\ 
  \Gamma_3^2\mapsto  \frac{\Gamma_1\Gamma_3\Gamma_4}{\Gamma_2},\ 
   \Gamma_4^2\mapsto \frac{\Gamma_1}{\Gamma_2\Gamma_3\Gamma_4} .$$
We expect that under these transformation, we have 
$M_{1,k}\mapsto M_{3,k}\mapsto M_{4,k}\mapsto M_{1,k}$ as well as $M_{2,k}\mapsto M_{2,k}$
for Macdonald operators, and 
$$ P_{\lambda_1,\lambda_2,\lambda_3,\lambda_4}\mapsto P_{\frac{\lambda_1+\lambda_2+\lambda_3+\lambda_4}{2},\frac{\lambda_1+\lambda_2-\lambda_3-\lambda_4}{2},\frac{\lambda_1-\lambda_2+\lambda_3-\lambda_4}{2},\frac{-\lambda_1+\lambda_2+\lambda_3-\lambda_4}{2}} $$
for Macdonald polynomials as well as $q$-Whittaker functions.

The fact that our operators/polynomials obey these symmetry relations is a highly non-trivial check of our
construction.

\subsubsection{M operators}

\begin{eqnarray*}
M_{1,k}&=&\sum_{\epsilon=\pm 1}\sum_{i=1}^4 \prod_{j\neq i} \frac{x_i^\epsilon x_j}{x_i^\epsilon x_j-1}
\frac{x_i^\epsilon}{x_i^\epsilon-x_j} x_i^{k\epsilon}\, \Gamma_i^{2\epsilon} \\
M_{2,k}&=&M_{1,k}^2-q^2 M_{1,k+1}\,M_{1,k-1}\\
{M}_{3,k}&=&
\sum_{\epsilon_1,\epsilon_2,\epsilon_3,\epsilon_4=\pm 1\atop \epsilon_1\epsilon_2\epsilon_3 \epsilon_4=-1}  
\prod_{1\leq i<j\leq 4} \frac{x_i^{\epsilon_i} x_j^{\epsilon_j}}{x_i^{\epsilon_i}x_j^{\epsilon_j}-1} 
\prod_{i=1}^4 x_i^{\frac{k\epsilon_i}{2}}\, \prod_{i=1}^4 \Gamma_i^{\epsilon_i}\\
{M}_{4,k}&=&
\sum_{\epsilon_1,\epsilon_2,\epsilon_3,\epsilon_4=\pm 1\atop \epsilon_1\epsilon_2\epsilon_3 \epsilon_4=1}  
\prod_{1\leq i<j\leq 4} \frac{x_i^{\epsilon_i} x_j^{\epsilon_j}}{x_i^{\epsilon_i}x_j^{\epsilon_j}-1} 
\prod_{i=1}^4 x_i^{\frac{k\epsilon_i}{2}}\, \prod_{i=1}^N \Gamma_i^{\epsilon_i}\\
\end{eqnarray*}

The expected symmetry relations under both $\Z_2$ and $\Z_3$ automorphisms are easily checked.
For instance in the $\Z_2$ case, we see immediately that $M_{3,k}$ and $M_{4,k}$ are interchanged 
under the transformation
$(x_4,\Gamma_4)\to (x_4^{-1},\Gamma_4^{-1})$ while all other $x_i, \Gamma_i$ remain unchanged,
as the transformation amounts to changing $\epsilon_4\to -\epsilon_4$ in the summation.

In the general $D_N$ case, the $\Z_2$ symmetry of the Dynkin diagram that interchanges the two end-nodes $N-1$ and $N$ implies the symmetry under $x_N\mapsto 1/x_N$ while all other $x$'s remain unchanged, together with $\Gamma_N\mapsto \Gamma_N^{-1}$. It is clear that under this transformation, we have
$M_{N-1,k}^{(D_N)}\mapsto M_{N,k}^{(D_N)}$ and $M_{N,k}^{(D_N)}\mapsto M_{N-1,k}^{(D_N)}$
Indeed, eqs.\eqref{Dtwo} and \eqref{Dthree} are interchanged under the change of summation variable
$\epsilon_N\to -\epsilon_N$, which amounts exactly to $x_N\to1/x_N$ and $\Gamma_N\to \Gamma_N^{-1}$.

\subsubsection{Dual $q$-Whittaker functions}

\begin{eqnarray*}
\Pi_{0,0,0,0}&=&s_{0,0,0,0}\\
\Pi_{1,0,0,0}&=&s_{1,0,0,0}\\
\Pi_{2,0,0,0}&=&s_{2,0,0,0}+q^{-2}s_{1,1,0,0}+q^{-4}s_{0,0,0,0}\\
\Pi_{1,1,0,0}&=&s_{1,1,0,0}+q^{-2}s_{0,0,0,0}\\
\Pi_{3,0,0,0}&=&s_{3,0,0,0}+\frac{1+q^2}{q^4}s_{2,1,0,0}+q^{-6}s_{1,1,1,0}+
\frac{1+q^2+q^4}{q^8}s_{1,0,0,0}\\
\Pi_{2,1,0,0}&=&s_{2,1,0,0}+q^{-2}s_{1,1,1,0}+
\frac{1+q^2}{q^4}s_{1,0,0,0}\\
\Pi_{1,1,1,0}&=&s_{1,1,1,0}+
q^{-2}s_{1,0,0,0}\\
\Pi_{4,0,0,0}&=&s_{4,0,0,0}+\frac{1+q^2+q^4}{q^6}s_{3,1,0,0}+\frac{1+q^4}{q^8}s_{2,2,0,0}+
\frac{1+q^2+q^4}{q^{10}}s_{2,1,1,0}\\
&&+q^{-12}(s_{1,1,1,1}+s_{1,1,1,-1})+\frac{(1+q^4)(1+q^2+q^4)}{q^{12}}s_{2,0,0,0}\\
&&+\frac{(1+q^4)(1+q^2+q^4)}{q^{14}}s_{1,1,0,0}+\frac{1+q^4+q^8}{q^{16}}s_{0,0,0,0}\\
\Pi_{3,1,0,0}&=&s_{3,1,0,0}+q^{-2}s_{2,2,0,0}+
\frac{1+q^2}{q^{4}}s_{2,1,1,0}+q^{-6}(s_{1,1,1,1}+s_{1,1,1,-1})\\
&&+
\frac{1+q^2+q^4}{q^6}s_{2,0,0,0}+\frac{1+q^2+2q^4}{q^8}s_{1,1,0,0}+\frac{1+q^4}{q^{10}}s_{0,0,0,0}\\
\Pi_{2,2,0,0}&=&s_{2,2,0,0}+
q^{-2}s_{2,1,1,0}+q^{-4}(s_{1,1,1,1}+s_{1,1,1,-1})\\
&&+
q^{-4}s_{2,0,0,0}+\frac{1+q^2+q^4}{q^6}s_{1,1,0,0}+\frac{1+q^4}{q^{8}}s_{0,0,0,0}\\
\Pi_{2,1,1,0}&=&s_{2,1,1,0}+q^{-2}(s_{1,1,1,1}+s_{1,1,1,-1})+
q^{-2}s_{2,0,0,0}+\frac{1+q^2}{q^4}s_{1,1,0,0}+q^{-6}s_{0,0,0,0}\\
\Pi_{1,1,1,\epsilon}&=&s_{1,1,1,\epsilon}+q^{-2}s_{1,1,0,0}+q^{-4}s_{0,0,0,0}\\
\Pi_{\frac{1}{2},\frac{1}{2},\frac{1}{2},\frac{\epsilon}{2}}&=&
s_{\frac{1}{2},\frac{1}{2},\frac{1}{2},\frac{\epsilon}{2}}\\
\Pi_{\frac{3}{2},\frac{1}{2},\frac{1}{2},\frac{\epsilon}{2}}&=&s_{\frac{3}{2},\frac{1}{2},\frac{1}{2},\frac{\epsilon}{2}}
+q^{-2}s_{\frac{1}{2},\frac{1}{2},\frac{1}{2},-\frac{\epsilon}{2}}\\
\Pi_{\frac{5}{2},\frac{1}{2},\frac{1}{2},\frac{\epsilon}{2}}&=&s_{\frac{5}{2},\frac{1}{2},\frac{1}{2},\frac{\epsilon}{2}}+q^{-2}s_{\frac{3}{2},\frac{3}{2},\frac{1}{2},\frac{\epsilon}{2}}+\frac{1+q^{2}}{q^4}s_{\frac{3}{2},\frac{1}{2},\frac{1}{2},-\frac{\epsilon}{2}}+\frac{1+q^2}{q^6}s_{\frac{1}{2},\frac{1}{2},\frac{1}{2},\frac{\epsilon}{2}}\\
\Pi_{\frac{3}{2},\frac{3}{2},\frac{1}{2},\frac{\epsilon}{2}}&=&s_{\frac{3}{2},\frac{3}{2},\frac{1}{2},\frac{\epsilon}{2}}+q^{-2}s_{\frac{3}{2},\frac{1}{2},\frac{1}{2},-\frac{\epsilon}{2}}+\frac{1+q^2}{q^4}s_{\frac{1}{2},\frac{1}{2},\frac{1}{2},\frac{\epsilon}{2}}
\end{eqnarray*}
for $\epsilon=\pm 1$.

The expected symmetry relations under the $\Z_3$ Dynkin automorphism are easily checked using the explicit expressions for Weyl-invariant D-type Schur functions \eqref{DNschur}. As an illustration, the reader can check that $\Pi_{2,0,0,0}\mapsto \Pi_{1,1,1,-1}$, as a consequence of $s_{2,0,0,0}\mapsto s_{1,1,1,-1}$,
$s_{1,1,0,0}\mapsto s_{1,1,0,0}$ and $s_{0,0,0,0}\mapsto s_{0,0,0,0}$ under the transformation \eqref{mapxD}. The $\Z_2$ symmetry is simply the covariance of $\Pi$ under $\epsilon\to-\epsilon$.

\end{appendix}

\bibliographystyle{alpha}

\bibliography{refs}
\end{document}